\documentclass[fleqn]{article}
\usepackage[left=2.5cm,right=2.5cm,top=1.5cm,bottom=1.5cm]{geometry}
\usepackage{float}
\usepackage{graphicx} % Required for inserting images
\usepackage{amsfonts}
\usepackage{amsmath}
\usepackage{lipsum}
\usepackage{longtable}
\usepackage{tabularx}
\usepackage{tabu}
\usepackage{mathtools}
\usepackage{stmaryrd}
\usepackage{amssymb}
\usepackage{amsthm}
\usepackage{tikz}
\usepackage{nameref}
\usepackage{multirow}
\usepackage{subcaption}

\usetikzlibrary{positioning, arrows.meta}
\usepackage{authblk} % for authors and affiliation
\usepackage[
authordate,
bibencoding=auto,
strict,
backend=biber,
maxcitenames=2,
maxbibnames=10,
uniquelist=false,
uniquename=init,
sorting=debug]{biblatex-chicago} 
\usepackage{newtxtext,newtxmath}

\AtEveryBibitem{%
  \clearfield{month}%
  \clearfield{isbn}%
  \clearfield{issn}%
  \clearfield{url}%
  \clearfield{Accessed}% <- removes "Accessed" date
  \clearfield{urldate}% <- removes "Accessed" date
  \clearfield{doi}%
  \clearfield{note}% <- removes "Publisher: ..." or other comments
  \clearfield{publisher}%
}

\newcounter{remark}
\newcommand{\remark}[1]{%
  \refstepcounter{remark}%
  \paragraph{Remark~\theremark}\label{#1}%
}

\newcommand{\RemarkRef}[1]{Remark~\ref{#1}}

\renewbibmacro*{name:andothers}{%
  \ifboolexpr
    {
      test {\ifnumequal{\value{listcount}}{\value{liststop}}}
      and
      test \ifmorenames
    }
    {
      \ifnumgreater{\value{liststop}}{1}
        {\finalandcomma}
        {}%
      \andothersdelim
      \bibstring{andothers}%
    }
    {}%
}

\newtheorem{theorem}{Theorem}[section]
\newtheorem{corollary}{Corollary}[section]

\newcounter{assumpA}
\renewcommand{\theassumpA}{A\arabic{assumpA}}

\newcommand{\assumpA}[1]{%
  \refstepcounter{assumpA}% advance counter and make ref labelable
  \textbf{(\theassumpA)}\label{#1}\,%
}

\newcommand{\assumpref}[1]{\ref{#1}}
\newcommand{\Assumpref}[1]{(\ref{#1})}  % for parentheses version

\title{Strategic Optimization of Bus Systems with Stochastic Ridership}

\author[1,*]{Haoran Zhao}
\author[1]{Andres Fielbaum}
\affil[1]{\small School of Civil Engineering, University of Sydney}
\affil[*]{\small Corresponding author email address: haoran.zhao@sydney.edu.au}\affil[*]{\small 
Corresponding author address: The University of Sydney, NSW 2006, Australia}
\date{}
\begin{document}
\maketitle
\begin{abstract}
  In global metropolitan areas, public transport benefits from bus systems. Bus design widely applies theoretical models, which typically assume static ridership. However, ridership randomness exists, and lack of attention to it might lead to wrong design decisions. Moving beyond static ridership, we extend the traditional single-line model with stochastic ridership. Further, to optimize the bus system under various ridership randomness, we introduce a hybrid model that combines conventional buses (CBs) and flexible buses (FBs). In both models, we apply a continuous approximation approach and find that: 1) Bus capacity increases with ridership randomness in the extended single-line model; 2) The hybrid model exhibits a binary state: either CBs serve the ridership with rejections under low ridership randomness, or CBs serve the majority while FBs serve the minority with rejections under high ridership randomness. Our findings establish a link between bus system design and ridership randomness, contributing to a more adaptive and efficient public transport framework.
\end{abstract}

\noindent\textbf{Keywords:}
Stochastic ridership; Bus system design; Flexible buses; Continuous approximation

\section{Introduction}
\subsection{Motivation}
Public transport supports the daily lives of millions of residents in a city, and buses are typically a backbone component. Accurately representing the travel demand is essential for planning and improving the design of the bus system \parencite{yu_analysing_2017}. Theoretical models have a long history of providing strategic insights to determine how to design bus systems for different demand patterns. These models are usually based on simplified representations of the demand and network to produce analytical results.

Previous theoretical models typically assume travel demand without randomness (e.g., \parencite{kim_integration_2014,sipetas_continuous_2021,calabro_adaptive_2023,fielbaum_relationship_2024}), i.e., the demand volumes are known a priori and deterministic. This assumption limits our understanding of public transport design \parencite{horcher_review_2021}, as ridership can be actually quite random \parencite{song_limits_2010}. Let us illustrate this using smartcard data from Canberra, Australia. When trips are aggregated into five-minute intervals, the trip counts vary significantly, even considering the same origin-destination pair and time period. This is shown in Figure \ref{fig:Canberra}, where we observe:

\begin{enumerate}
    \item Every curve is far from being stable, even though we depict lapses of 5 minutes. That is, Figure \ref{fig:Canberra} suggests that the demand changes significantly every 5 minutes, which can mostly be explained by randomness.
    \item Different working days also present relevant differences.
\end{enumerate}

\begin{figure}[htbp]    \includegraphics[width=\linewidth]{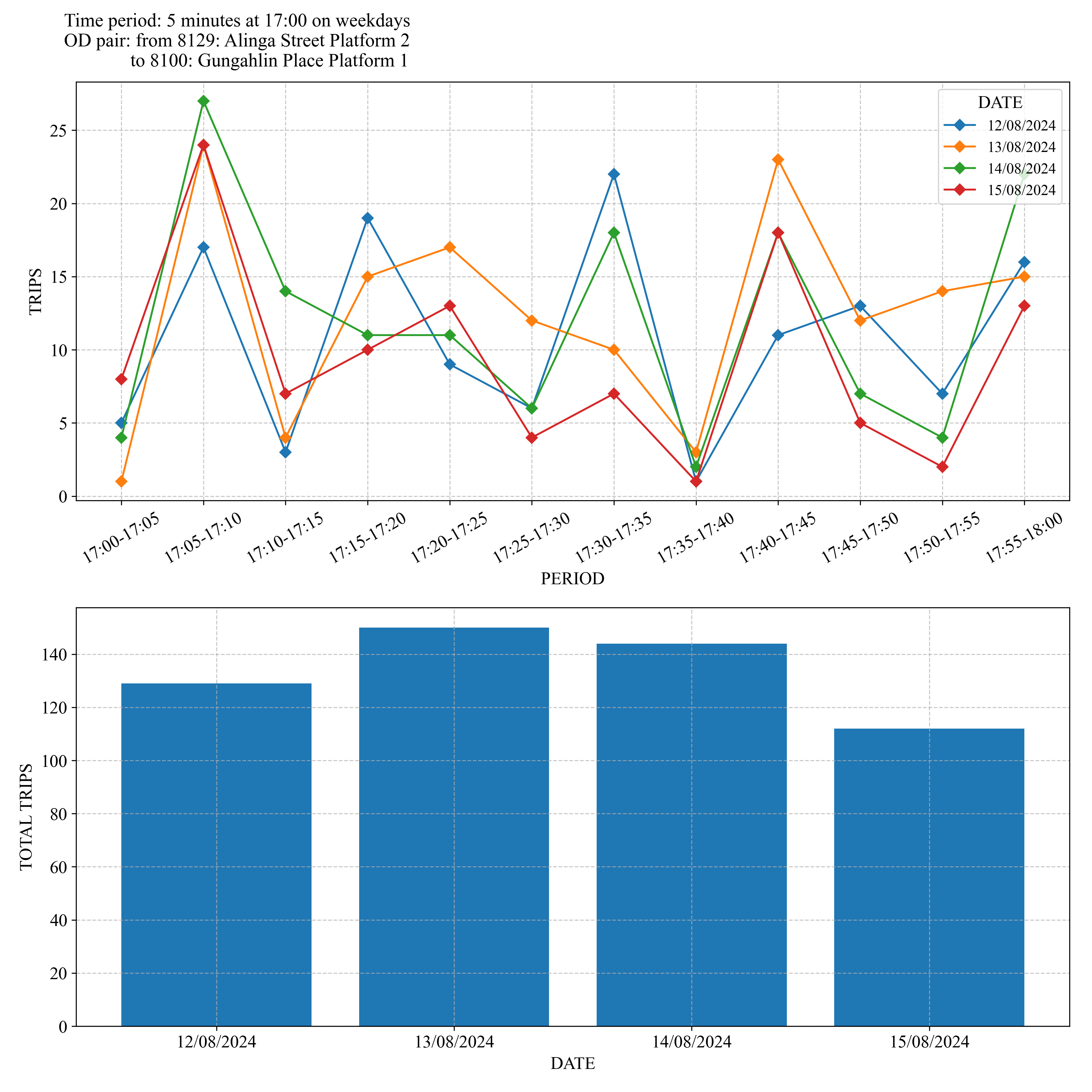}
    \caption{OD trips from Alinga Street Platform 2	to Gungahlin Place Platform 1 in Canberra, Australia}
    \label{fig:Canberra}
\end{figure}

The lack of attention to ridership randomness might lead to incorrect design decisions. Particularly, if buses are designed for the average ridership, they might get overcrowded when a random surge occurs, potentially leading to some users being unable to board the vehicle, a situation that occurs frequently in many urban public transport systems \parencite{drabicki_mitigating_2023}. Therefore, it is critical to expand previous strategic models to account for this stochasticity.

Not only is this topic relevant, but it is particularly timely in terms of how to address it: the ability of online coordination between users and vehicles can be a potent tool to reduce the impacts of randomness. Specifically, if the public transport system gathers online information about the origin and destination of the passengers, it can re-route some vehicles from lines that have fewer passengers than expected to others facing random demand.

\subsection{Related works}
\subsubsection{Online coordination in public transport}
Online coordination between demand and supply has been studied to introduce novel mobility into public transport, including demand-responsive bus systems, and related modes \parencite{shaheen_shared_2016, narayanan_shared_2020}.
% such as \parencite{vansteenwegen_survey_2022} carpooling and ridepooling \parencite{siuhi_opportunities_2016}. Mobility design develops into two types: stop-based service \parencite{vallee_maximizing_2017, t_perera_simulation_2018, melis_real-time_2022}, and door-to-door service \parencite{alonso-mora_-demand_2017, navidi_comparison_2018, song_optimized_2024, wang_demand-responsive_2024}. 
Both passengers and operators benefit from advanced information by online coordination \parencite{melis_real-time_2022}, and the effectiveness of such coordination has been demonstrated for demand-responsive buses to outperform traditional buses in rural areas \parencite{coutinho_impacts_2020} and areas with high peak/off-peak demand variability \parencite{fielbaum_improving_2024, perera_resurgence_2020}.

%Recognizing these advantages, on-demand mobility can be integrated into traditional public transport. For an extensive review of the integration of on-demand services into public bus systems, see \textcite{vansteenwegen_survey_2022}. 

When on-demand services are analyzed within a public transport network, a well-studied alternative is to use them as feeders for the first- and last-mile of public transport \parencite{badia_feeder_2020, melis_integrated_2024, aravind_equity_2024}. \textcite{badia_feeder_2020} determines whether fixed-route service or door-to-door service is more competitive as first- and last-mile solution. They emphasize that unless the door-to-door feeder service significantly reduces the operators' costs, it could not outperform the fixed-route service.
\textcite{melis_integrated_2024} incorporated on-demand buses into a high-frequency metro network by formulating an on-demand bus routing problem. They found that the integration system outperforms during peak hours, while only on-demand buses operate during off-peak hours.
\textcite{aravind_equity_2024} highlight that operating on-demand service as a feeder for the fixed-route public transport improves the accessibility significantly, supported by the case study of Morristown, Tennessee, US. For a thorough review, the reader is referred to \textcite{vansteenwegen_survey_2022}.

%Integrating on-demand buses into public transport can provide a more effective solution to the first- and last-mile problem, when demand is consistent at peak hour and the integration reduces costs. Such integration can also enhance the cost efficiency and accessibility of public transport.

Beyond the first- and last-mile problem, the integration of on-demand services into public transport has also been analyzed at the citywide scale. At the bus planning stage, \textcite{kumar_planning_2022} determined which existing bus lines should be removed, designed the optimal frequencies for the remaining lines, and optimized the fleet size of mobility-on-demand across various zones. Based on experiments conducted on the Sioux Falls network, they found that congestion could be reduced under such a hybrid design.  

The integration of on-demand mobility with traditional public transport is promising for both first- and last-mile services and citywide applications. However, few studies have examined analytically the joint optimization of fleet size and vehicle capacity, which are essential components of transport planning. Moreover, all of these papers require a specific coordination between passengers and vehicles, as opposed to letting the passengers take the first available vehicle, whose simplicity provides advantages for users and operators \parencite{fielbaum_coordination_2025}.

To the best of our knowledge, this is the first study where online coordination is analyzed at the planning level and within a given set of lines in order to address the inherent randomness in bus demand.

\subsubsection{Continuous approximation and single-line model}

Research on bus systems can be done at a strategic, tactical, or operational level. We follow a strategic approach, which enables us to study critical questions such as fleet sizing and frequencies. More specifically, we adopt a \textit{continuous approximation approach},
which employs continuous functions and assumptions to estimate system costs. This high-level method, widely applied in public transport design, offers analytical tractability and captures insightful general trends in public transport \parencite{daganzo_increasing_1987, calabro_adaptive_2023}.

\textcite{ansari_advancements_2018} provide a comprehensive review of the application of continuous approximation in logistics and transportation systems. Network design has been one of the applications of continuous approximation models. For example, the method has been applied to bus service design \parencite{calabro_adaptive_2023}, bus routing \parencite{ellegood_continuous_2015}, and bus spacing \parencite{pandey_bus_2024}. 

Several papers have employed such an approach to design public transport including the integration of on-demand mobility with conventional services. \textcite{calabro_adaptive_2023} adopted continuous approximation to design an adaptive transit system, which deploys fixed-routes and demand-responsive feeders. They found that such integration significantly improves users costs and reduce access time to main truck services. \textcite{silva_-demand_2022} proposed a family of continuous approximation models to predict system performance, which provides information for fleet sizing and driving scheduling of on-demand services.

A specific type of continuous approximation is the so-called \textit{single-line model}. It focuses on the design of a single bus line rather than the overall network structure. This model analyzes one line in isolation, allowing researchers to gain transparent insights, derive closed-form expressions, and obtain analytical conclusions. Early analytical models developed in the 1970s established the fundamental relationships between demand, frequency, and bus capacity \parencite{mohring_optimization_1972, jansson_simple_1980}. In a single-corridor setting, \textcite{mohring_optimization_1972} found increasing returns to scale in public bus services, as higher demand leads to higher frequency and shorter waiting times. The optimal frequency was found to be proportional to the square root of demand. Extending this analysis, \textcite{jansson_simple_1980} jointly optimized frequency and bus capacity, showing that while frequency remains proportional to the square root of demand, bus capacity is approximately—but not exactly—proportional to it within a given demand range.

Subsequent research has largely built upon these principles to address more complex operational and behavioral settings. Recently, the single-line model has been applied to study optimal frequency and bus size \parencite{jara-diaz_optimal_2017}, bus stop spacing \parencite{zhang_optimization_2020,fielbaum_relationship_2024}, vehicle scheduling \parencite{c_tang_vehicle_2019}, and pricing \parencite{ramos_fare_2023}. Acknowledging its established role in the literature, we employ a single-line model based on continuous approximation to optimize the fleet size and bus capacity for an integrated system of flexible and conventional buses.

\subsection{Contribution}
In this study, we apply a continuous approximation approach to design public transport services accounting for random ridership. We start by expanding the traditional single-line model that determines the optimal frequency and bus capacity for a given stochastic demand level. The introduction of random ridership requires introducing the possibility of \textit{rejections} in our models, i.e., users who cannot board the bus due to overcrowding. We optimize the system's frequency and provide analytical results on the impact of randomness. 

To solve the challenges posed by the random nature of demand, we propose a \textit{hybrid} system, where some buses run as usual within a line, and others might get rerouted between different lines. This is studied through an expansion of the single-line model where several lines run in parallel, inspired by the models of \textcite{chang_multiple_1991,fielbaum_lines_2021}.

The major contributions of this paper are as follows:
\begin{enumerate}
    \item We extend the single-line model by including stochastic ridership and consequentially rejections due to overcrowding. We present analytical results regarding the changes in the optimal design variables as the demand level becomes more random. 
    \item Through our theoretical hybrid model and analytical results, we show that route flexibility and online coordination enable effective management of ridership randomness.
\end{enumerate}

The paper is organized as follows. Section 2 introduces the extended single-line model under stochastic ridership, along with its theoretical analysis. Section 3 develops and analyzes a hybrid bus system that combines conventional and flexible buses. Section 4 presents numerical experiments on both the extended single-line model and the hybrid model to validate the theoretical findings. Section 5 discusses the role of assumptions and their impact on the results. Section 6 concludes the paper by summarizing the main results and discussing directions for future research.

\section{Single-line model}
In this section, we extend the traditional single-line model by considering stochastic ridership. Subsequently, we analyze the behavior of the optimal design under varying levels of ridership randomness. 

\subsection{Model and costs}
Let us start with a straight bus line with only two stops, which is adapted from classic single-line models \parencite{mohring_optimization_1972, jansson_simple_1980, kraus_discomfort_1991}. This two-bus-stop simplification concentrates the model on the demand instead of other line structure factors that may affect the optimal design. We design the frequency $f$ and the bus capacity $K$ to serve a ridership $X$ between the two stops. Passengers board at one bus stop and alight at the other in one direction with no stops in between. We denote the cycle time by $T$.

Let us first recall the traditional single-line model that assumes $X$ to be fixed and aims to minimize the sum of users' and operators' costs. Users' costs are calculated by multiplying ridership $X$ by waiting costs $\frac{\pi_w}{2f}$ plus in-vehicle costs $\frac{\pi_v T}{2}$, where $\pi_w$ and $\pi_v$ are the values of waiting time and in-vehicle time for passengers. Operators' costs follow the formula by \textcite{jansson_simple_1980}. The fleet size is given by $s = f T$, and $s$ is multiplied by the cost of acquiring and operating each bus $\gamma_0 + \gamma_1 K$, where $\gamma_0$ and $\gamma_1$ represent the capacity-independent and capacity-dependent cost components, respectively. The optimization model includes non-negativity and capacity constraints and leads to closed-form expressions for the optimal frequency $f^*$ and bus capacity $K^*$ \parencite{jara-diaz_towards_2003}. 

\begin{align}
    &\text{min} \quad (\pi_w\frac{1}{2f}+\pi_v\frac{T}{2})X + fT(\gamma_0 + \gamma_1 K)\\ 
    &\text{subject to} \nonumber\\
    & \quad f,K\geq 0\\
    & \quad fK \geq X
\end{align}

We extend the traditional single-line model by incorporating stochastic ridership, which considers rejections. The ridership is no longer deterministic but follows a distribution with a cumulative distribution function $\mathbf{F}(x) = P(X\leq x)$ and a probability density function $\mathbf{f}(x)$. If the demand happens to be too large and does not fit in the buses, then some passengers are rejected: This can be interpreted as traditional denied boarding due to overcrowding, or through some sort of online coordination prior to going to the bus stop. Note that we make the following assumption \assumpA{ASP:rej_pax}: the rejected passengers leave the system, as opposed to waiting for the next availabe vehicle. The impact of this assumption is further discussed in Section~\ref{subsec:AspDis}. All notations in the extended single-line model are as in Table~\ref{tab:Single-line}.
\begin{table}[htbp]
    \centering
    \begin{tabularx}{\textwidth}{ccX}
    \hline
    Type & Note & Description \\
    \hline
    Inputs&$f$ & The frequency of buses (veh/h)\\
    &$K$ & The capacity of buses (seats)\\
    \hline
    Ridership & $X$ & The ridership, which is deterministic in the traditional single-line model and stochastic in the extended single-line model (pax/h)\\
    &$E(X_B)$& The mean of served bus ridership (pax/h)\\
    &$E(X_R)$& The mean of rejected ridership (pax/h)\\
    \hline
    &$\gamma_0$ & Capacity-independent term of operators' costs  (US\$)\\
    &$\gamma_1$ &  Capacity-dependent term of the operators' costs  (US\$/seat)\\
    Exogenous&$\pi_w$ &  The value of waiting time of passengers  (US\$/h)\\
    Constants&$\pi_v$ &  The value of in-vehicle time of passengers (US\$/h)\\
    &$T$& The cycle time of buses. (h)\\
    &$\Pi_R$& The penalty for rejection. (US\$/pax)\\
    \hline
    &$\mathbf{F}$& The cumulative distribution function of the ridership \\ 
    Probability&$\mathbf{f}$& The probability density function of the ridership \\ 
    &$P_R$& The probability of rejection \\ 
    \hline
    Capacity &$Q$& The systematic capacity for the single-line model\\ 
    \hline
    \end{tabularx}
\caption{Notations in single-line models}
\label{tab:Single-line}
\end{table}

We now explain the details of the single-line model with stochastic ridership. For now, we take $f$ and $K$ as given to write all the components of the cost function, so we can then optimize their values. The bus system rejects passengers if the ridership exceeds the capacity of the bus system, and the rejection probability is
\begin{equation}
    P_R = 1 - \mathbf{F}(Q)
\end{equation}
where $P_R$ is the probability of rejections and $Q$ is the system capacity, which is the frequency times the capacity, i.e. 
\begin{equation}
    Q = fK
\end{equation}

Costs in the extended single-line model comprise expected users', operators', and rejection costs. For each served passenger, waiting and in-vehicle costs are the same as in the traditional model. So are operators' costs. The key differences are the ridership and rejection costs. The ridership consists of served and rejected passengers. The mean of the served ridership is 
\begin{align}\label{Eq:ExpectedServedRidershipSingleLine}
    E(X_B)= \int_0^Q x \mathbf{f}(x)dx + Q P_R
\end{align}
where the first term in Eq. \eqref{Eq:ExpectedServedRidershipSingleLine} is the expected served ridership when it is fully supplied, and the second term is the served ridership when buses are overloaded. The expected value of the rejected ridership is 
\begin{align}
    E(X_R) = \int_Q^{\infty} (x-Q) \mathbf{f}(x)dx
\end{align}

The rejection costs are the multiple of the rejection penalty $\Pi_R$ and the rejected ridership. We design the frequency $f$ and bus sizes $K$ using an optimization model to minimize expected total costs $G(f, K)$ subject to non-negativity constraints.

\begin{align}
    &\text{min} \quad G(f, K)= (\pi_w\frac{1}{2f}+\pi_v\frac{T}{2})E(X_B) + fT(\gamma_0 + \gamma_1 K)+\Pi_R E(X_R) \label{eq:single_obj}\\ 
    &\text{subject to} \nonumber\\
    &\quad f,K \geq 0
\end{align}

Note that there is no capacity constraint here as rejections are admissible. Given the random nature of the demand, a strict capacity constraint might render the problem unfeasible, or force the system to provide capacity for the worst-case scenario, which can be unwise if the maximum feasible demand level is too large.

\subsection{Theoretical analysis}
\label{sec:theo_single}
In this subsection, we analyze the results of the first-order conditions to determine the impacts of ridership randomness on the system's design. To do this, let us formalize $\sigma$, the standard deviation of the ridership, as the crucial parameter to analyze. More specifically, we assume that the ridership has a constant mean and its distribution can be parametrized by $\sigma$, with $\mathbf{F}_{\sigma}$ denoting the cumulative distribution function for a given standard deviation (for example, the ridership might follow a log-normal distribution whose mean is given). We denote by $f^*(\sigma), K^*(\sigma)$ the resulting optimal frequency and capacity given $\sigma$. 

In order to proceed with the formal analysis, we need to consider the following assumption.

\textbf{Assumption} \assumpA{ASP:pax_dis}: At the right tail of the distribution, the cumulative distribution function decreases with the rise in the standard deviation of ridership. Mathematically, if $\sigma_1 > \sigma_2$, then for all $x$ sufficiently large, $\mathbf{F}_{\sigma_1}(x) < \mathbf{F}_{\sigma_2}(x)$.

This assumption might look technical but is quite natural. It follows from the fact that an increase in the standard deviation spreads the distribution more widely, leading to a heavier right tail where density values become relatively larger for higher values of $x$. This assumption holds for various distributions, including the normal, gamma, and beta distributions, as illustrated in Figure~\ref{fig:Distributions}.

Assumption~\assumpref{ASP:pax_dis} is used to analyze the optimal solutions under the first-order conditions and to determine the impact of the standard deviation of the ridership on the optimal frequency and capacity. Based on Assumption~\assumpref{ASP:pax_dis} and the derived optimal solutions, we establish the main result of this section, which is formally stated in Theorem \ref{Thm:Capacity}.

\begin{figure}[htbp]
\includegraphics[width=\linewidth]{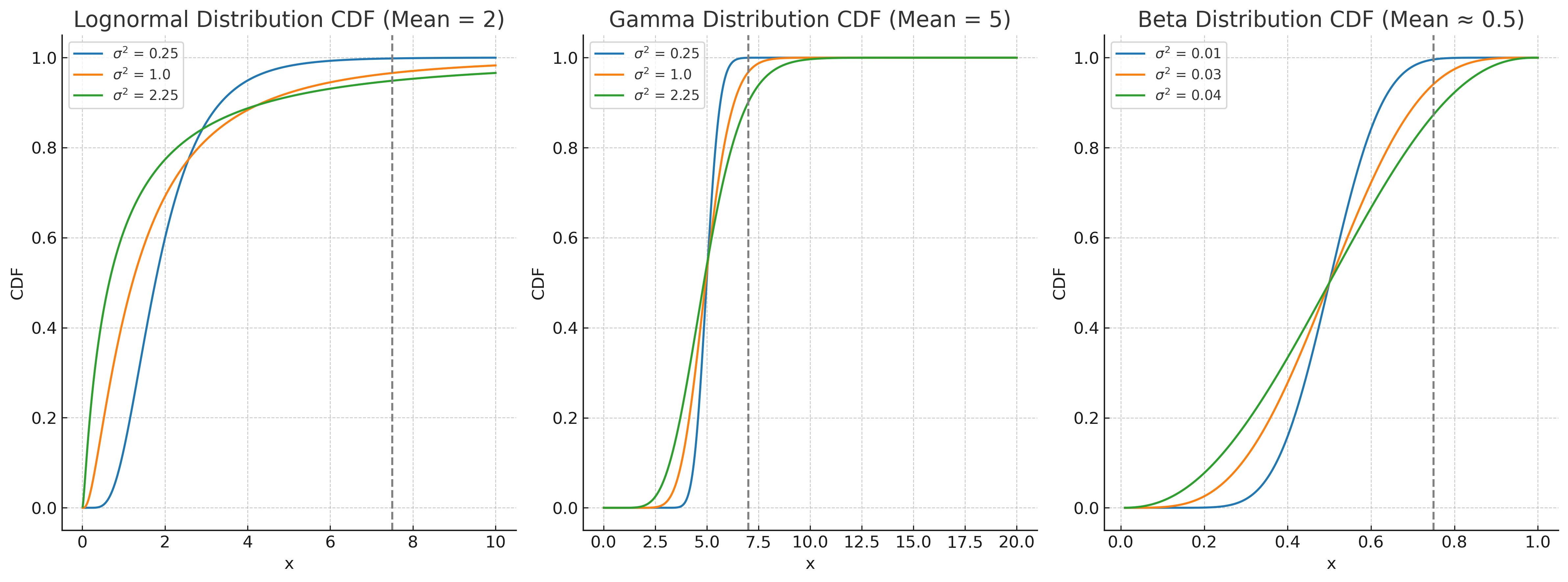}
\caption{The cumulative distribution function  of Lognormal, Gamma, and Beta distribution under several standard deviations}
\label{fig:Distributions}
\end{figure}

\begin{theorem}\label{Thm:Capacity}
    Assume a fixed expected ridership and a large rejection penalty. Then, the optimal bus capacity ($K^*$) increases, while the optimal frequency ($f^*$) remains nearly constant, as the ridership standard deviation ($\sigma$) increases. Formally, $\frac{df^*}{d\sigma} \approx 0, \frac{dK^*}{d\sigma} > 0$.
\end{theorem}

Theorem \ref{Thm:Capacity} describes that capacity, rather than frequency, accommodates ridership variability. The proof of Theorem \ref{Thm:Capacity} is provided in the appendix~\ref{apdx:comp_thm_cap}, but let us now discuss its core. The optimal frequency is shown to depend linearly on the square root of the served ridership $\sqrt{E(X_B)}$, although the served ridership itself is a function of frequency and bus capacity. This relationship explains why the optimal frequency remains nearly constant under a significant rejection penalty, as the system is designed to serve nearly the entire ridership ($E(X_B) \approx E(X)$).

From an intuitive perspective, ridership variability implies a chance of having a greater ridership than the mean, which should lead to an increased total capacity of $fK$ given the large rejection penalty. Why should this be faced through $K$ rather than $f$? According to Mohring's classical model and its further variations, frequency represents a trade-off between users and operators \parencite{mohring_optimization_1972,jara-diaz_towards_2003}. Increasing frequency reduces users' waiting costs but raises operators' costs as a larger fleet is needed. As the expected ridership is constant, the way to face this trade-off is unaffected by the stochasticity. With a large rejection penalty and stochastic ridership, operators should provide extra capacity to avoid high rejection costs. So, in order to prevent rejecting users, the total capacity is adapted in the cheapest possible way, which is by increasing the vehicles' sizes, as illustrated in Figure~\ref{fig:scheme_extended_single}.

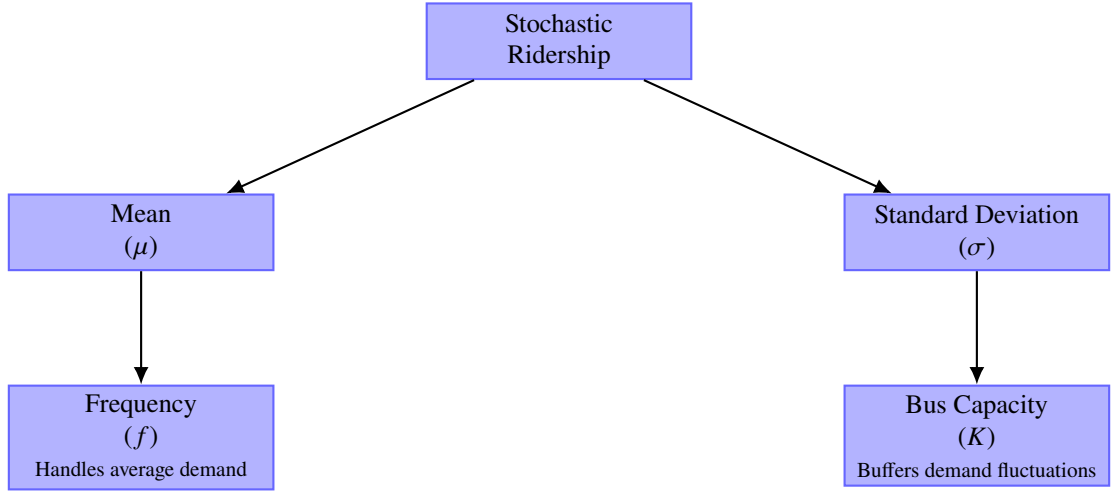
\begin{figure}[htbp]
\centering
\begin{tikzpicture}[
  node distance=1.5cm and 2cm,
  box/.style={rectangle, draw=blue!60, fill=blue!30, thick, minimum width=3.5cm, minimum height=1cm, align=center},
  arrow/.style={-{Latex[width=2mm]}, thick}
]

% Nodes
\node[box] (ridership) {Stochastic\\Ridership};
\node[box, below left=of ridership] (mean) {Mean\\($\mu$)};
\node[box, below right=of ridership] (std) {Standard Deviation\\($\sigma$)};
\node[box, below=of mean] (freq) {Frequency\\($f$)\\\footnotesize Handles average demand};
\node[box, below=of std] (capacity) {Bus Capacity\\($K$)\\\footnotesize Buffers demand fluctuations};

% Arrows
\draw[arrow] (ridership) -- (mean);
\draw[arrow] (ridership) -- (std);
\draw[arrow] (mean) -- (freq);
\draw[arrow] (std) -- (capacity);

\end{tikzpicture}
\caption{Relationship between stochastic ridership and the determination of frequency and bus capacity.}
\label{fig:scheme_extended_single}
\end{figure}

\section{Hybrid bus system} 
In the previous section, we extended the traditional single-line model to account for the randomness of the ridership. To accommodate stochastic ridership, flexible buses (FBs) that do not follow fixed lines can be an effective tool \parencite{ fielbaum_design_2024, wang_terminal-based_2024}. In this section, we develop a specific integration between FBs and conventional buses (CBs) to improve the system’s adaptability to stochastic ridership, leveraging their complementary strengths. Following prior papers applying continuous approximation \parencite{sipetas_continuous_2021, calabro_adaptive_2023}, we analyze the aggregation rather than passenger-by-passenger.

\subsection{Model}

We introduce a hybrid model that combines CBs and FBs within a designated city region, where passengers are taking feeder bus lines to board a train line; Note that the train line is not an explicit part of our model. We build a parallel bus line model based on the multi-line models proposed by \textcite{fielbaum_lines_2021, chang_multiple_1991}, in which bus lines are evenly spaced with identical horizontal intervals. This spatial model has the virtue of being simple enough to admit an analytical treatment, and at the same time it allows for the dynamics of the FBs. 

As illustrated in Figure~\ref{fig:Model}, CBs operate along $N$ fixed bus lines, following a vertical and parallel travel pattern. In contrast, FBs also travel vertically but follow flexible routes when rebalancing between bus stops, allowing them to adjust to demand fluctuations dynamically. Passengers are generated at origin stops (blue in Figure~\ref{fig:Model}), travel to their destination stops (orange in Figure~\ref{fig:Model}), and then leave the system. We assume no demand from destination stops to origin stops \assumpA{ASP:one_dir}. CBs and FBs collect passengers at their origins and transport them vertically to their destinations without stops. Each bus line is assumed to have two bus stops \assumpA{ASP:two_stop}. Upon finishing their vertical travels, FBs might be rebalanced, i.e., they might start their next trip from a different blue node. Note that the demand still follows Assumptions \Assumpref{ASP:rej_pax} and \Assumpref{ASP:pax_dis}. We introduce these extra assumptions to simplify the model; however, the main results are not constrained by them. This is further discussed in Section~\ref{subsec:AspDis}.
\begin{figure}[htbp]
    \centering
    \includegraphics[width = \textwidth]{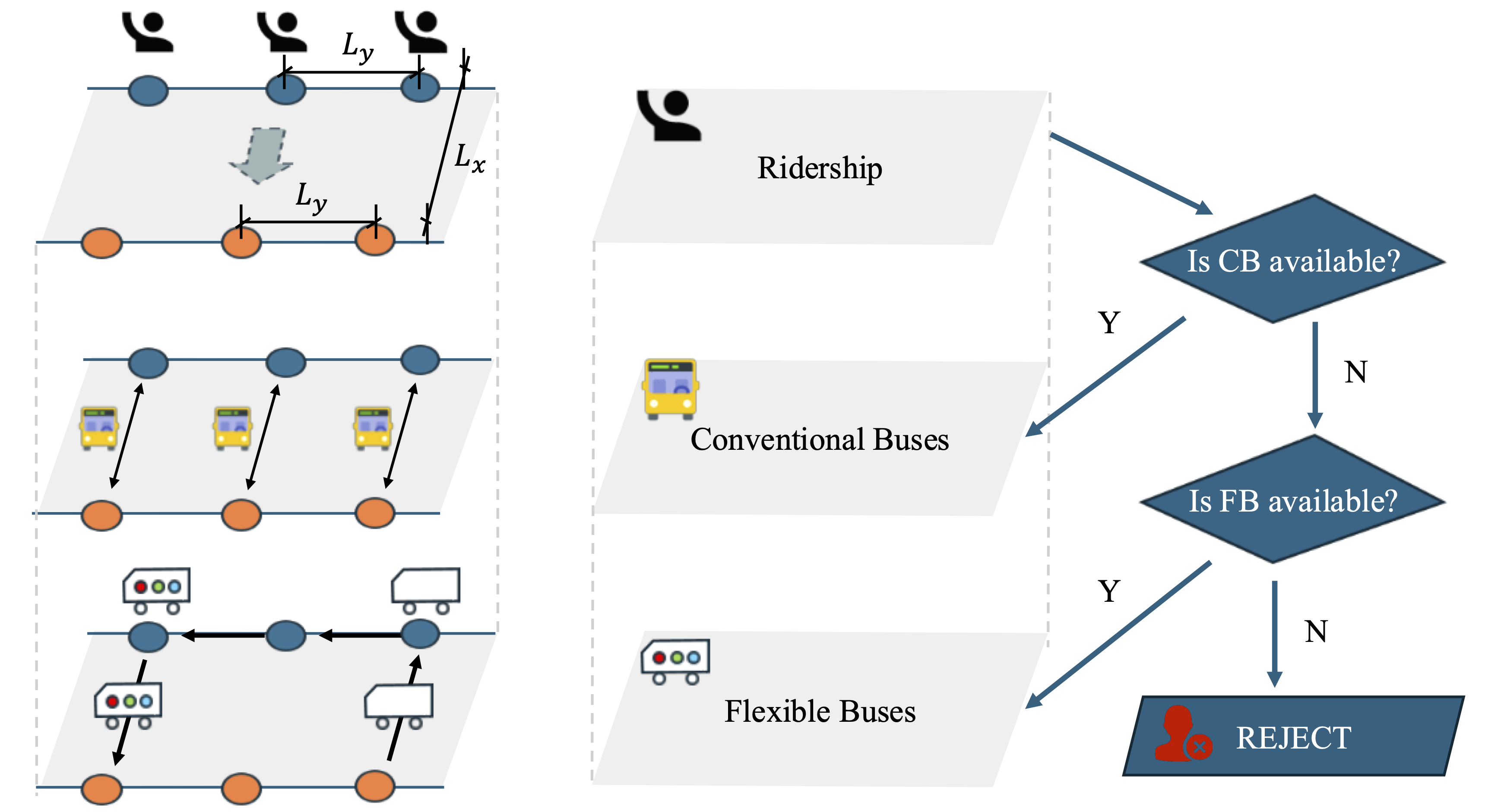}
    \caption{the network for CBs and FBs serving stochastic ridership (blue nodes: origin stops; orange nodes: destination stops) (Left) and the coordination between CBs and FBs (Right)}
    \label{fig:Model}
\end{figure}

Passengers inform their trips online, enabling the system to make efficient rebalancing decisions. The system first checks the availability of CBs. If no CBs are available, it then assesses the availability of FBs. If neither CBs nor FBs are available, passengers are rejected by online coordination. Instead of assuming deterministic ridership, our model incorporates stochastic ridership to account for its variance. Let $X$ denote the ridership of the whole system as a random variable following a probability distribution with cumulative distribution function $\mathbf{F}(x) = P(X\leq x)$. The ridership on each vertical line is assumed to be independently and identically distributed (i.i.d.). The coordination splits the ridership into CBs, FBs, and rejection ridership, denoted as $X_C$, $X_F$, and $X_R$; these three random variables fulfill $X_C+X_F+X_R=X$. 

We apply continuous approximation to estimate the costs of the hybrid bus system and formulate an optimization model to minimize the expected costs. The decision variables comprise frequency ($f_C$) and bus capacity ($K_C$) for CBs, as well as fleet size ($s_F$) and bus capacity ($K_F$) for FBs. All relevant notations are summarized in Table~\ref{Table:Notations_Hybrid}. In the subscript notation, $F$ represents FBs, $C$ represents CBs, $R$ represents rejections, $u$ represents users, and $o$ represents operators.
\begin{table}[htbp]
\caption{Notations of the hybrid bus system}
\label{Table:Notations_Hybrid}
\centering
\begin{tabularx}{\textwidth}{llX}
\hline
Type & Note & Description \\
\hline
\multirow{4}{*}{Design Variables} 
 & $f_C$ & Frequency of CBs (veh/h) \\
 & $K_C$ & Bus capacity of CBs (seats) \\
 & $s_F$ & Fleet size of FBs (veh) \\
 & $K_F$ & Bus capacity of FBs (seats) \\
\hline
\multirow{6}{*}{Ridership}
 & $X$ & Total transit ridership (CBs, FBs, and rejected) (pax/h) \\
 & $X_C$ & Ridership of CBs (pax/h) \\
 & $X_F$ & Ridership of FBs (pax/h) \\
 & $X_R$ & Rejected ridership (pax/h) \\
 & $E(X_C)$ & Expected ridership of CBs (pax/h) \\
 & $E(X_F)$ & Expected ridership of FBs (pax/h) \\
 & $E(X_R)$ & Expected rejected ridership (pax/h) \\
\hline
\multirow{7}{*}{Exogenous Constants}
 & $N$ & Number of bus lines \\
 & $L_x$ & Vertical distance of the region and trip length (km) \\
 & $L_y$ & Horizontal distance between bus stops (km) \\
 & $V$ & Speed of CBs and FBs (km/h) \\
 & $\gamma_0$ & Capacity-independent operational cost of CBs and FBs (US\$) \\
 & $\gamma_1$ & Capacity-dependent operational cost of CBs and FBs (US\$/seat) \\
 & $\pi_w$ & Value of waiting time (US\$/h) \\
 & $\pi_v$ & Value of in-vehicle time (US\$/h) \\
 & $\Pi_R$ & Penalty for rejection (US\$/pax) \\
\hline
\multirow{3}{*}{Probabilities}
 & $P_C$ & Probability of being underserved by CBs \\
 & $P_R$ & Probability of a passenger being rejected \\
 & $P_{j,i}^B$ & Probability of FBs rebalanced from origin $j$ to $i$ \\
\hline
\multirow{6}{*}{Costs}
 & $c_{F,w}$ & Waiting costs of FBs (US\$) \\
 & $c_{F,v}$ & In-vehicle costs of FBs (US\$) \\
 & $c_{F,u}$ & Users' costs of FBs (US\$) \\
 & $c_{F,o}$ & Operators' costs of FBs (US\$) \\
 & $c_{C,u}$ & Users' costs of CBs (US\$) \\
 & $c_{C,o}$ & Operators' costs of CBs (US\$) \\
\hline
\multirow{4}{*}{Detours \& Capacity}
 & $D$ & Detours of FBs (km) \\
 & $Q_{CB}$ & Capacity of CBs (pax/h) \\
 & $Q_{FB}$ & Capacity of FBs (pax/h) \\
 & $Q_{HY}$ & Capacity of the hybrid model (pax/h) \\
\hline
\end{tabularx}
\end{table}

\subsection{Conventional buses}
The CB model considers a fixed-route bus system, where CBs operate vertically without intermediate stops. We assume that all lines are identical. The total costs consist of two components: users' costs ($c_{C,u}$) and operators' costs ($c_{C,o}$). We first focus on users' costs.
\begin{align}
    & c_{C,u} = \Bigl(\frac{\pi_w }{2f_C} +  \frac{ \pi_v L_x }{V}   \Bigr) E(X_C)\label{ur_cost_CB_new}
\end{align}
where $f_C$ is the frequency of CBs. $\pi_w$ and $\pi_v$ represent the waiting and in-vehicle time values for CBs passengers, respectively. $L_x$ denotes the travel distance, and $V$ is the average speed of CBs. For each passenger, the waiting and in-vehicle costs remain the same as in the single-line model. However, the difference lies in the expected ridership for CBs, denoted as $E(X_C)$, which is further discussed in Section~\ref{Sec: ridership}.

The operators' costs function is analogous to the single-line case extended to a multi-line system, where it equals the product of fleet size and the cost of acquiring and using a bus. Given identical bus lines, the operators' costs are determined by the sum of the product for all identical bus lines as
\begin{align}
    & c_{C,o} = N(f_C\frac{2L_x}{V})(\gamma_0 + \gamma_1 K_C)\label{op_cost_CB_new}
\end{align}
where $N$ is the number of bus lines, and $K_C$ is the bus capacity of CBs. $\gamma_0$ and $\gamma_1$ represent the fixed and per-capacity cost components of CBs operations, respectively.

\subsection{Flexible buses}
In this section, we first introduce the operation of FBs and then formulate users', operators' and total costs of FBs.

\subsubsection{FBs operation}
FBs pick up passengers at an origin stop and then search for stops underserved by CBs. FBs travel toward the targeted underserved stop, picking up passengers along the way. After collecting passengers at the target stop, FBs travel directly to the destination stop without intermediate stops. FBs operate empty from destination stops to origin stops and then begin the next service cycle. As shown in Figure~\ref{fig:Rebalance}, FBs are rebalanced through the following sequential search process:
\begin{enumerate}
    \item \textbf{Checking the current bus stop:}  
    FBs first assess the current bus stop. 
    \begin{itemize}
        \item If CBs underserve this stop, FBs are assigned to it, and the search terminates. 
    \end{itemize}
     
    \item \textbf{Evaluating following bus stop pairs:}  
     If the nearest stop is adequately served, FBs proceed to evaluate the following bus stop pair under these conditions:  
    \begin{itemize}
        \item If both stops in the pair are underserved, they receive equal assignment probability, and FBs are allocated to one of them, terminating the search.
        \item If one stop is underserved while the other is not, FBs are allocated to the underserved stop, and the search terminates.
        \item If both stops are adequately served, the search continues to the next following pair, repeating the evaluation process until an underserved stop is found. 
    \end{itemize}
\end{enumerate}
Once the rebalancing search terminates, FBs proceed to the identified underserved origin along the \textit{detour} route and pick up passengers at bus stops before traveling to their vertically aligned destination stops. 
\begin{figure}[htbp]
    \centering
    \includegraphics[width = \textwidth]{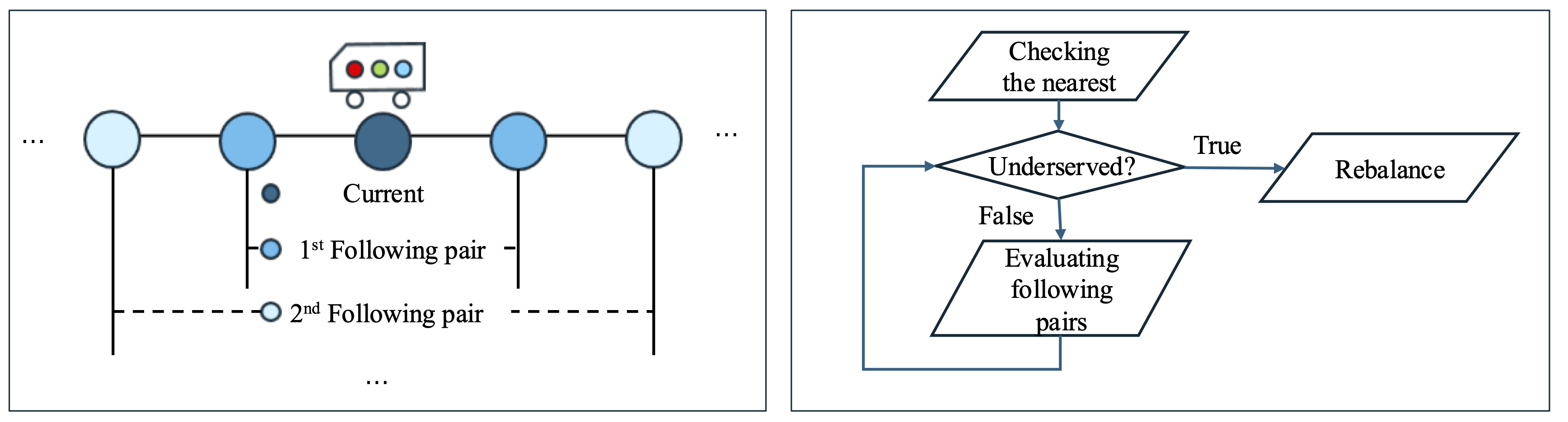}
    \caption{Diagram of bus stops in the rebalance (Left) and flow chart of rebalance search (Right)}
    \label{fig:Rebalance}
\end{figure}

The detour length is determined by the search process, which depends on the probability of being underserved. Since all lines are identical, bus stops have the same probability of being underserved by CBs. Bus stops are considered underserved by CBs when the ridership exceeds CBs capacity. The probability of bus stops being underserved by CBs, denoted as $P_C$, is 
\begin{align}
    P_C &= \mathbf{F}(\frac{X}{N}\geq K_C f_C)
\end{align}
The probability of FBs being rebalanced from origin $j$ to $i$ results from $P_C$. This relation is formally described in Eq.~\eqref{reb-1} - Eq.~\eqref{reb-2}. Eq.~\eqref{reb-1} represents that the current bus stop is underserved. Eq.~\eqref{reb-2} describes the probability that FBs are assigned to the origin of $k$ intervals away from the current origin. This implies the current origin and origins of $\{1 \cdots k-1\}$ intervals are adequately supplied by CBs, which can be described as $(1-P_C)(1-P_C)^{2(k-1)}$.
\begin{align}
    P_{i,i}^B &= P_C \label{reb-1}\\
    P_{i,i+k}^B &=  (1-P_C)^{2k-1}(P_C-\frac{P_C^2}{2})\label{reb-2}
\end{align}
where $P_{j, i}^B$ is the probability of FBs being rebalanced from origin $j$ to origin $i$. The distance between origins $j$ and $i$, denoted by $D_{j,i}$, is $|j - i|L_y$, where $L_y$ is the distance for the intervals between vertically aligned bus stops. The expected detour distance $D$ of FBs is 
\begin{align}
    D &= \sum_{k \in \mathbb{Z}} P_{j, k}^B D_{j, k}\label{eq:detour-tmp1}\\
    &= 2  \sum_{k=1}^{\infty}  P_{j,j+k}^B k L_y \label{eq:detour-tmp2}\\
    &= L_y \frac{1-P_C}{P_C (2-P_C)} \label{eq:detour}
\end{align}
Eq.~\eqref{eq:detour-tmp1} represents the expected detour distance from a bus stop $j$. Note that there is a necessary approximation here: since far-away detours are unlikely to take place (which is further discussed in \RemarkRef{Rmk:CB-FB}), we sum not only among the $N$ lines but to infinity, since the tail of this sum is expected to yield a low value. Note that this implies the following assumption \assumpA{ASP:Homo_line}: all parallel lines are exactly the same, which is why we obtain a single $D$ (independent of $j$)  in Eq. \eqref{eq:detour}. We discuss this assumption in Section \ref{subsec:AspDis}.

The symmetry of the region transforms Eq.~\eqref{eq:detour-tmp1} into Eq.~\eqref{eq:detour-tmp2}. The transition from Eq.~\eqref{eq:detour-tmp2} to Eq.~\eqref{eq:detour} follows from standard algebraic simplification. 

\remark{Remark 1:} \label{Rmk:CB-FB} According to Eq.~\eqref{eq:detour}, CBs influence the operation of FBs by affecting their detours, mathematically through $P_C$. When CBs provide insufficient capacity for public transit ridership, as $P_C$ approaches 1, FBs run with short detours to complement CBs in terms of capacity. Conversely, when CBs provide abundant capacity as $P_C$ is close to 0, FBs take long detours to pick up passengers underserved by CBs. Note that this is an inefficient scenario that will never be optimal. The costs of operating FBs exceed the limited benefit of serving  very few underserved passengers.
\footnote{Even when the detours are upper-bounded, if FBs cannot identify an underserved stop within the detour, they would continue to serve as CBs rather than wasting capacity on detours to serve the demand already covered by existing capacity. This ``meaningless-detours'' situation prevents the integration of FBs into the optimal design.}

\subsubsection{Cost of flexible buses}
The total costs of FBs consist of the users' and operators' costs. Let us begin with users' costs, which can be divided into waiting and in-vehicle costs. First, we establish the waiting costs $c_{F,w}$:
\begin{align}
    c_{F,w} &= \pi_w  \omega \frac{N(2L_x + D)}{s_F V}E(X_F) \label{eq:wait-cost-FBs} 
\end{align}
Waiting costs are the product of the waiting time, expected FBs ridership $E(X_F)$, and the time value $\pi_w$. Passengers wait at bus stops while FBs take detours to pick them up. Consequently, the expected waiting time is half the average headway between FBs\footnote{Note that FBs do not have a constant headway because they do not follow a fixed route. But as we are computing the average waiting time, we can use the average headway between vehicles, using the traditional formula that headway is equal to cycle time divided by fleet size.}. $E(X_F)$ denotes the expected FBs ridership. The estimated waiting time is $\omega \frac{N(2L_x + D)}{s_F V}$, where $\omega$ represents the waiting time fraction for the FBs \parencite{saidi_long-term_2016}, capturing the stochastic nature of FB arrivals. $\frac{2L_x + D}{V}$ is the cycle time of FBs, and $\frac{s_F}{N}$ is the average fleet size serving passengers at a bus stop. In-vehicle costs are determined by the in-vehicle time of FBs ridership multiplied by the time value $\pi_v$, given by
\begin{align}
    & c_{F,v} = \frac{\pi_v  (L_x + \frac{D}{2}) E(X_F)}{V}   \label{in-veh-cost-FBs}
\end{align}
where $\frac{L_x + \frac{D}{2}}{V}$ is the travel time, and $\frac{D}{2}$ is the average detour. $E(X_F)$ represents the expected ridership of FBs, which is defined in Section~\ref{Sec: ridership}.

The operators' costs of FBs $c_{F,o}$ are modeled the same as those of CBs. The total operators' cost is determined by the fleet size $s_F$ and the cost per bus, which has a linear dependency on the capacity $K_F$. $\gamma_0$ and $\gamma_1$ are the exogenous costs of the operation and acquisition of FBs. 
\begin{align}
    & c_{F,o} = s_F(\gamma_0 + \gamma_1 K_F) \label{eq:op-cost-FBs}
\end{align}

\subsection{Optimization model}
In this subsection, we introduce the rejection, the expected ridership, and the optimization model to design the frequency ($f_C$) and bus capacity ($K_C$) for CBs as well as the fleet size ($s_F$) and bus capacity ($K_F$) for FBs.
\subsubsection{Rejection and ridership} \label{Sec: ridership}
Passengers could be rejected, if the capacity of the hybrid model is insufficient to accommodate the ridership. The capacity of the hybrid system ($Q_{HY}$) is the sum of CBs ($Q_{CB}$) and FBs capacity ($Q_{FB}$), denoted as
\begin{align}
     Q_{HY} &= Q_{FB} + Q_{CB} \\
     &=K_F  \frac{s_F V}{\omega(2L_x + D)} (1 + \frac{D}{L_y}) + NK_Cf_C \label{eq:hybrid_cap}
\end{align}
The first term is the hourly capacity of FBs, which is the bus capacity of FBs ($K_F$) times the estimated number of FBs in service ($s_F V/\omega(2L_x + D)$) and the estimated number of bus stops visited by FBs ($1 + D/L_y$), where $2L_x + D$ is the cycle length. Although detours lead to longer cycle times, they enables FBs to pick up more passengers. The second term is the hourly CBs capacity. For $N$ bus lines, the hourly capacity is the bus capacity $K_C$ times $f_C$ the frequency. The probability of rejection $P_R$ refers to the condition where the ridership exceeds the hybrid capacity, given by
\begin{align}
    P_R = \mathbf{F}(X \geq Q_{HY})
\end{align}

In the hybrid model, the ridership consists of CBs, FBs, and rejected ridership. The expected CBs ridership is 
\begin{align}
    &E(X_C) = \int_0^{Q_{CB}} x \mathbf{f}(x) dx + P_C Q_{CB} \label{eq:hybrid_CB_ridership} 
\end{align}
where $\int_0^{Q_{CB}} x \mathbf{f}(x) dx$ is the CBs ridership given that CBs fully serve the ridership. But when CBs underserve the ridership, the served ridership of CBs is $Q_{CB}$, and the probability of underserving is $P_C$. Following the same idea, the expected FBs and rejected ridership are
\begin{align}
    &E(X_F) = \int_{Q_{CB}}^{Q_{HY}} (x - Q_{CB}) \mathbf{f}(x) dx + P_R Q_{FB} \label{eq:hybrid_FB_ridership} \\   
    & E(X_R) = \int_{Q_{HY}}^{\infty} (x -Q_{HY}) \mathbf{f}(x) dx \label{eq:hybrid_REJ_ridership}
\end{align}

\subsubsection{Optimization of hybrid model}

In the hybrid model, we optimize two sets of variables: (1) frequency $f_C$ and bus capacity $K_C$ for CBs, and (2) fleet size $s_F$ and bus capacity $K_F$ for FBs. The goal is to minimize the expected hybrid model costs consisting of users', operators' and rejection costs introduced above, formulated as Eq.~\eqref{eq:objective}. The optimization model includes non-negativity constraints for FBs and CBs.

\begin{align}
    &\text{min} \quad  \underbrace{\left( \pi_w \omega \frac{N(2L_x + D)}{s_F V} + \frac{\pi_v  (L_x + \frac{D}{2}) }{V} \right) E(X_F)} _{\text{Users' Costs of FBs}}
    + \underbrace{s_F(\gamma_0 + \gamma_1 K_F)}_{\text{Operators' Costs of FBs}}  \nonumber\\
    &\quad +\underbrace{ \Bigl(\frac{\pi_w}{2f_C} +   \frac{\pi_v L_x}{V} \Bigr)E(X_C)}_{\text{Users' Costs of CBs}}
    +\underbrace{N(2f_C\frac{L_x}{V} )(\gamma_0 + \gamma_1 K_C)}_{\text{Operators' Costs of CBs}}
    +\underbrace{E(X_R)\Pi_{REJ}}_{\text{Rejection Costs}}\label{eq:objective}\\ 
    &\text{subject to:} \nonumber\\
    &\quad f_C, K_C, s_F, K_F \geq 0 \label{con:non-neg}
\end{align}

\subsection{Theoretical Analysis }
In this section, we introduce Theorem~\ref{Thm:low-var-case} and Theorem~\ref{Thm:high-var-case}, which explain that the hybrid model exhibits a binary state. According to \RemarkRef{Rmk:CB-FB}, CBs influence FBs by detours. We distinguish two cases.
\begin{enumerate}
    \item  Low variance case: We prove that as the standard deviation of ridership approaches zero, the probability of bus stops being underserved by CBs tends to zero. This demonstrates that CBs are more efficient in low-variance scenarios where demand remains stable.
    \item  High variance case: We then prove that when the standard deviation of ridership is sufficiently large, the fleet sizes of CBs and FBs are positive. In high-variance settings, FBs are introduced into the system to leverage its online coordination and route flexibility advantages, adapting to fluctuating demand more efficiently than CBs.
\end{enumerate}

\begin{theorem} \label{Thm:low-var-case}
    As the standard deviation of ridership approaches zero ($\sigma \to 0$), the probability of being underserved by CBs tends to zero ($P_C \to 0$).
\end{theorem}
\begin{proof}
      When the variance converges to 0, the cumulative distribution function approaches a step function as 
      \begin{align}
          \mathbf{F}(X) = \begin{cases}
              0, X < E(X)\\
              1, X \geq E(X)\\
          \end{cases}
      \end{align} and the random variable for ridership tends to be a constant $E(X)$. As we design the system to serve the ridership, $Q_{HY} \geq E(X)$, then $\mathbf{F}(Q_{HY})$ is approaching 1, and $P_R$ converges to 0. We divide the structure of the hybrid model into two cases:
      \begin{enumerate}
          \item If CBs do not serve the expected ridership, it means that  $P_C \to 1$ (because there is almost no variance). This suggests that the demand is served solely by FBs. But thanks to Eq. \eqref{eq:detour}, we observe that in this case the detour approaches zero. That is to say, FBs operate actually as CBs in this context. In other words, CBs serve the ridership when the variance tends to zero.
          \item If CBs serve the expected ridership, as $P_C \to 0$, FBs tend to take infinite detours to find passengers who are underserved by CBs (Eq. \eqref{eq:detour}). Therefore, the costs to incorporate FBs into the hybrid bus system tend to infinity. Therefore, CBs serve the ridership when the variance tends to zero.
          \end{enumerate}

\end{proof}

Therefore, when the standard deviation of ridership approaches 0, there are no FBs in the hybrid model. To conclude, in low variance scenarios, CBs demonstrate their advantages of serving stable demand. From Theorem~\ref{Thm:low-var-case} and 
 \RemarkRef{Rmk:CB-FB}, we obtain the corollary as follows.

\begin{corollary} \label{Coro:low-var-1}
    When the standard deviation of ridership is low ($\sigma < \epsilon$), detour lengths of FBs are enormous, and the optimal hybrid design excludes FBs.
\end{corollary}
Including FBs with extensive detours is not optimal, as it imposes excessive waiting costs on FB passengers. However, detour lengths decrease as ridership randomness increases, making FBs efficient under high ridership variability and integrated into the optimal hybrid model. This relationship is formally established in Theorem \ref{Thm:high-var-case}.

\begin{theorem}\label{Thm:high-var-case}
    If $\sigma$ is sufficiently large, the optimal fleet size of both CBs and FBs is strictly positive.
\end{theorem}

 Theorem~\ref{Thm:high-var-case} describes that the increase in ridership randomness introduces FBs into the hybrid model. The proof of Theorem~\ref{Thm:high-var-case} is in the appendix~\ref{apdx:comp_thm_high}. In the proof, an increase in system randomness amplifies the rejection penalty while simultaneously reducing detours. This leads to a more efficient rebalancing. Consequently, passengers can wait for a rebalanced FB in a reasonable time rather than being rejected by CBs. The optimal frequency of CBs still linearly depends on the square root of the served CBs ridership. The FBs rebalanced from other vertical lines enable the system to absorb high demand with a relatively greater probability compared to the low-variance case. As a result, the integration of FBs contributes positively to the hybrid bus system under high ridership randomness.
 
From Corollary \ref{Coro:low-var-1}, note that the enormous detours prevent FBs from the optimal design when the ridership randomness is low.  The length of detours decreases, and the capacity of FBs increases with ridership randomness. When the standard deviation of ridership exceeds a threshold, the non-negative constraint on the fleet size of FBs is no longer effective, and FBs join the optimal hybrid design. This suggests that not only FBs appear but also that they appear in a discontinuous fashion.

The optimal hybrid design either excludes flexible buses (FBs) entirely or incorporates a substantial number of them, as demonstrated by Theorem~\ref{Thm:low-var-case} and Theorem~\ref{Thm:high-var-case}. The inclusion of only a minimal number of FBs to address ridership stochasticity is inefficient. This finding provides insights into the strategic integration of CBs and FBs, offering valuable insights into the design principles of hybrid bus systems.

\section{Numerical analysis and discussion}
We implement a grid search across ridership randomness, given the non-convexity and non-linearity of the extended single-line model and the hybrid model. We assume that the ridership follows the logarithmic normal distribution to disentangle the mean and the variance. Finally, we discuss the behavior of the optimal fleet size and bus capacity with respect to the standard deviation of ridership. The exogenous parameters for the numerical experiments are set as in Table~\ref{tab:constants}. We define the values of $N$, $L_x$, $L_y$ and $\Pi_R$, and the remaining parameter values follow \parencite{fielbaum_transit_2018}. Note that we set a high rejection penalty to prioritize passenger service.

\begin{table}[htbp]
\centering
\caption{Parameters of the extended single-line model and the hybrid model}
\label{tab:constants}
\begin{tabularx}{\textwidth}{l c c X}
\hline
Type & Note & Value & Description \\
\hline
\multirow{2}{*}{Stochastic Ridership} 
 & $\mu$ & 100 & Expected ridership (pax/h) \\
 & $\sigma$ & \{0,10,\dots,100\} & Standard deviation of ridership in the single-line model (increments of 10, pax/h) \\
 &         & \{0,10,\dots,200\} & Standard deviation of ridership in the hybrid model (increments of 10, pax/h) \\
\hline
\multirow{10}{*}{Exogenous Constants}
 & $N$ & 10 & Number of origin–destination stops in the region \\
 & $L_y$ & 0.5 & Distance interval between bus stops (km) \\
 & $L_x$ & 20 & Vertical distance of the region (km) \\
 & $V$ & 30 & Average speed of CBs and FBs (km/h) \\
 & $T$ & $4/3$ & Cycle time of CBs (h) \\
 & $\Pi_R$ & 30 & Penalty cost of rejection (US\$/pax) \\
 & $\pi_w$ & 4.44 & Monetary value of waiting time (US\$/h) \\
 & $\pi_v$ & 1.48 & Monetary value of in-vehicle time (US\$/h) \\
 & $\gamma_0$ & 4.02 & Constant component of operator cost (US\$) \\
 & $\gamma_1$ & 0.89 & Variable cost coefficient (US\$/seat) \\
 & $\omega$ & 0.5 & The fraction of estimated waiting time for FBs \\
\hline
\end{tabularx}
\end{table}

\subsection{The extended single-line model}
We analyze the behavior of the optimal frequency and bus capacity as the standard deviation of ridership increases. Ridership is modeled with a mean of 100 (pax/h) and a standard deviation ranging from 10 to 100 (pax/h) in increments of 10 (pax/h). 

\begin{figure}[htbp]
    \centering    \includegraphics[width = \textwidth]{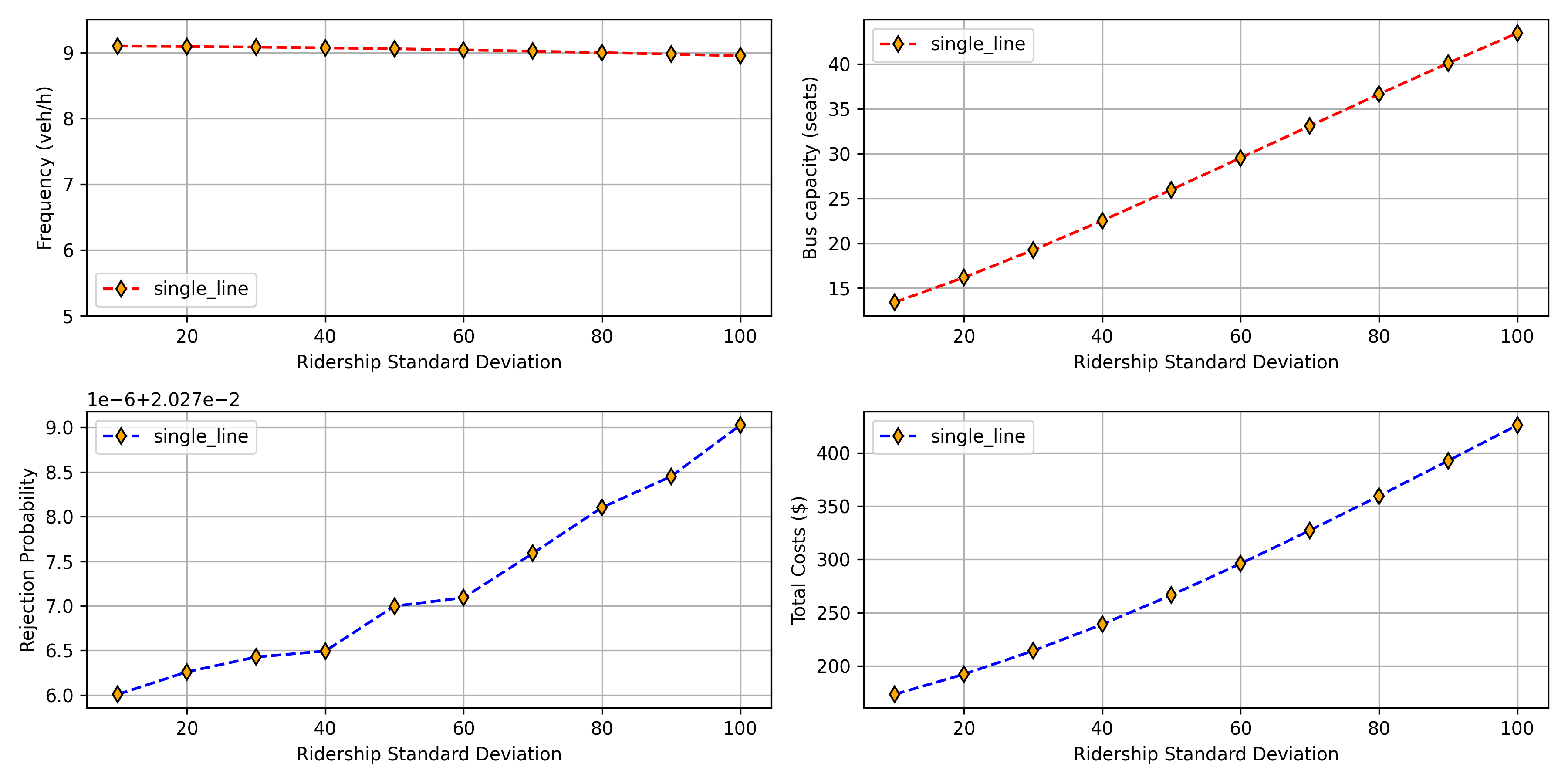}
    \caption{Optimal designed variables of the extended single-line model}
    \label{fig:OP-INPUT-Single}
\end{figure}

We solve the optimization model to obtain the optimal frequency $f$ and bus capacity $K$. The behavior of the optimal frequency and bus capacity in the extended single-line model is shown in Figure~\ref{fig:OP-INPUT-Single}. As the standard deviation of the ridership increases, the optimal frequency stays almost constant while the optimal capacity significantly increases, consistent with Theorem \ref{Thm:Capacity}. Given a substantial penalty, the rejection probability is low but increases slightly as the standard deviation of ridership rises. The behavior of the optimal frequency and bus capacity corresponds to the theoretical analysis findings in the extended single-line model. As discussed in the theoretical analysis, the optimal frequency increases with expected ridership. Passengers prefer higher frequency to reduce travel costs, but this increases operators' costs. Thus, frequency is the trade-off between users' and operators' costs. However, in our experiments, expected ridership remains constant while ridership randomness increases. Greater bus capacity absorbs ridership peaks and expands the system's ability to handle variability. Additional numerical experiments with an expected demand of 1,000 passengers per hour (presented in the appendix~\ref{appd:more_exps}) confirm the robust behavior of the optimal design variables across different demand scenarios.

\begin{figure}[htbp]
    \centering    \includegraphics[width = 0.8\textwidth]{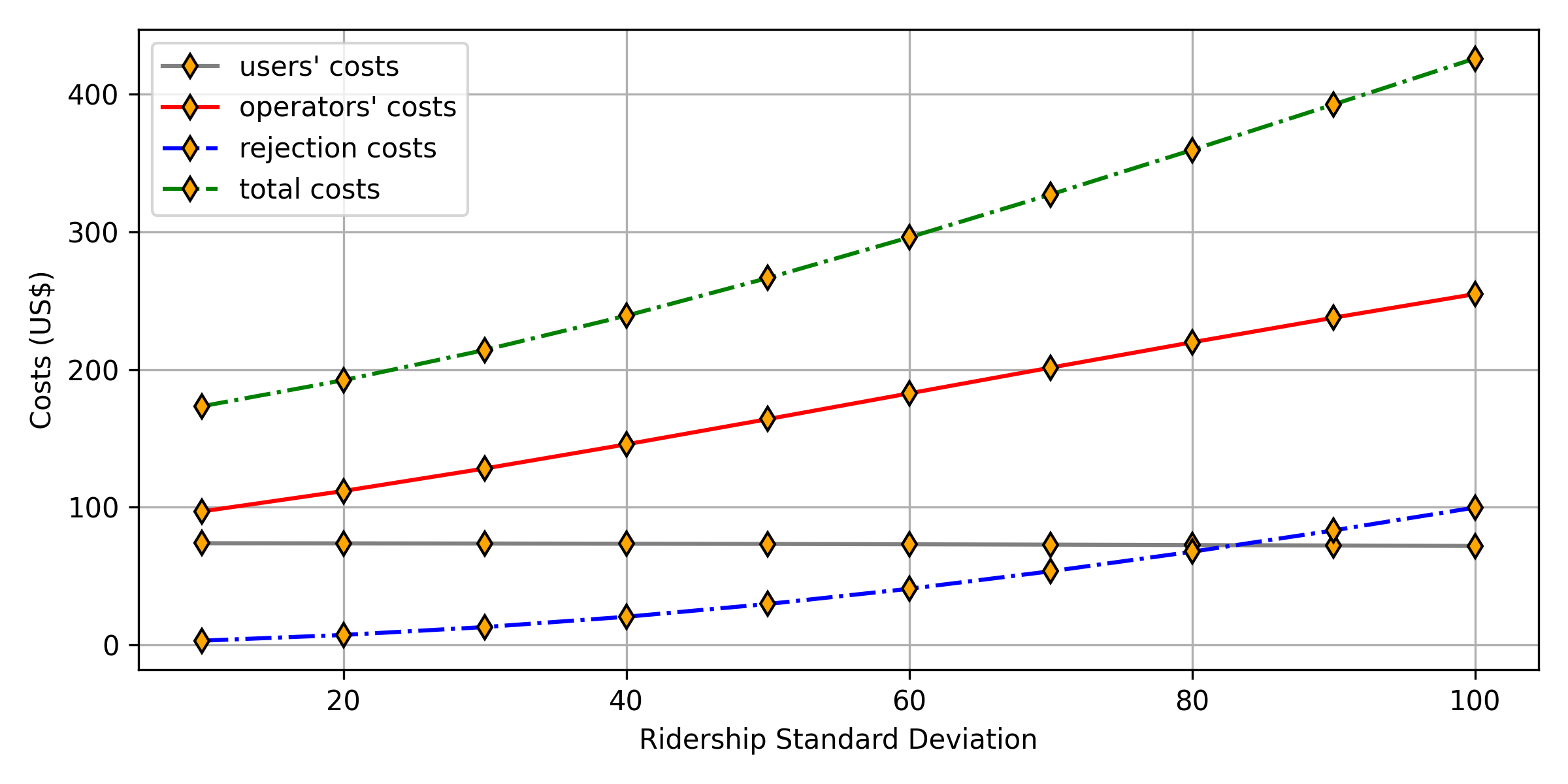}
    \caption{Breakdown of the total costs under optimal design of the extended single-line model.}
    \label{fig:SingleCosts}
\end{figure}

Users', operators' and rejection costs are illustrated in Figure~\ref{fig:SingleCosts}. When the standard deviation of ridership increases, the optimal design increases the bus capacity to mitigate the rising probability of rejection, leading to higher operators' costs. However, this adjustment only partially offsets the mismatch between the supply and stochastic demand, leaving rejections that continue to grow. Given that the mean ridership remains constant, the served ridership declines slightly, resulting in a slight decrease in the users' costs (note that these costs only consider the served passengers).

To conclude, given a significant rejection penalty, the optimal frequency serves the mean ridership, and the optimal bus capacity accommodates the ridership randomness. Through both theoretical and numerical analysis, our findings robustly highlight this knowledge gap in managing stochastic ridership variability and emphasize the role of capacity in addressing demand uncertainty.

\subsection{The hybrid model}
In the hybrid model, we analyze the behavior of the optimal design variables for CBs and FBs, when increasing the standard deviation of ridership from 0 to 200 (pax/h) in increments of 10 (pax/h) while maintaining a mean ridership of 100 (pax/h). 

As shown in Figure~\ref{fig:OP-INPUT-Hybrid}, when the standard deviation of ridership is low, CBs handle all of the demand, capitalizing on their advantage of providing consistent service capacity for stable demand, which is consistent with Theorem \ref{Thm:low-var-case}.

With increasing ridership randomness, FBs are integrated into the hybrid model to address this variability, as expected due to Theorem \ref{Thm:high-var-case}. That is, the hybrid model has a binary state: either CBs serve all the ridership, or CBs serve the majority and FBs serve the minority. This aligns with Theorem \ref{Thm:low-var-case} and Theorem \ref{Thm:high-var-case}.
\begin{figure}[htbp]
    \centering
    \includegraphics[width = \textwidth]{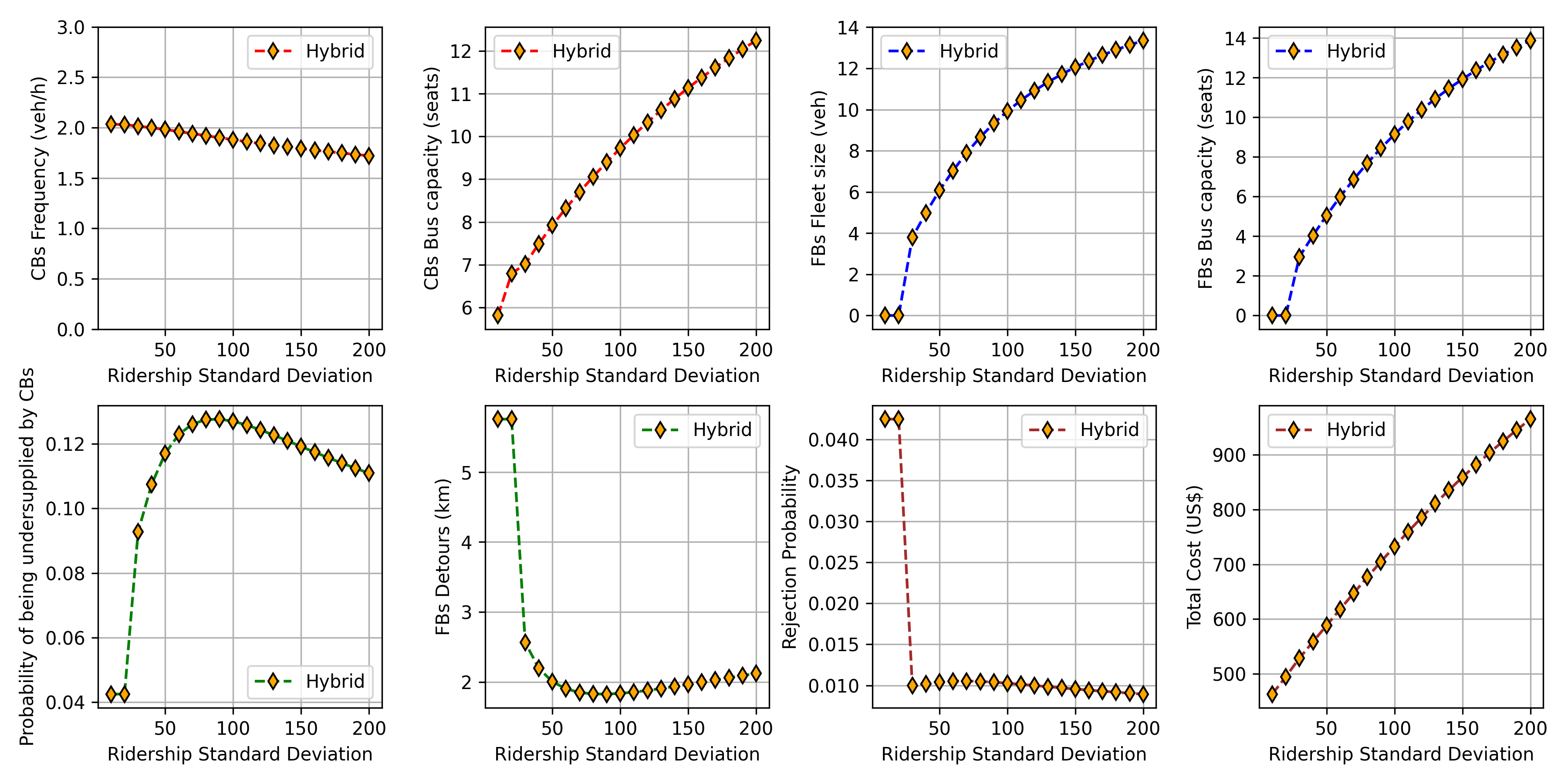}
    \caption{Optimal design variables for the hybrid model with mean ridership of 100 passenger per hour}
    \label{fig:OP-INPUT-Hybrid}
\end{figure}

 This transition is reasonable, as a higher standard deviation of ridership increases the probability of ridership deviating significantly from its expected value.  If we follow the idea that CBs are designed for the expected ridership in prior papers on deterministic ridership, the rejection probability will become high under such variability. Since the system is designed to serve ridership with a significant rejection penalty, the hybrid model employs a more adaptable capacity to accommodate stochastic ridership compared to the capacity for deterministic ridership. For this increased ridership, why do we integrate FBs into the hybrid model rather than a model consisting of CBs solely? Let us conduct a further analysis. 
 
 \begin{figure}[htbp]
    \centering
    \includegraphics[width = 0.8\textwidth]{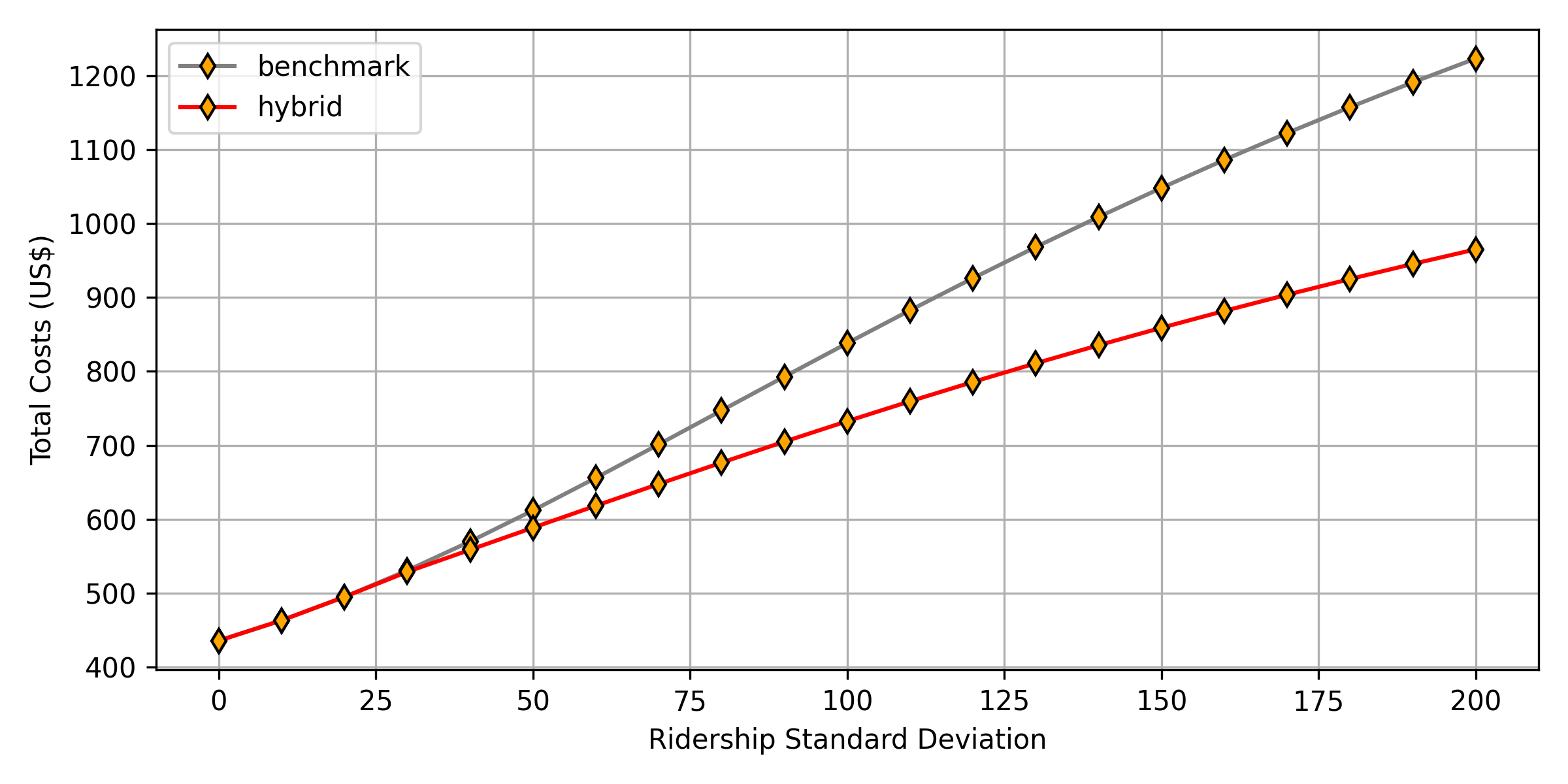}
    \caption{Comparison of the hybrid model to the benchmark model with mean ridership of 100 passenger per hour}
    \label{fig:OP-COM-Hybrid}
\end{figure}

As shown in Figure~\ref{fig:OP-COM-Hybrid}, the hybrid model outperforms the benchmark model, which comprises only CBs for $N$ bus lines. This superior performance stems from route flexibility. As the ridership randomness increases across bus stops, FBs take shorter detours, while CBs require greater capacity for all lines. Once the increase in capacity costs exceeds the detour costs, FBs are more economical than CBs. However, CBs remain necessary to ensure that detours remain within a reasonable range. In this case, under a high standard deviation of ridership, it is more economically efficient to incorporate FBs in a hybrid system where CBs serve a portion of the passenger demand to leverage economies of scale while FBs manage ridership randomness through route flexibility. Transcending the role of flexible mobility as a feeder \parencite{aravind_equity_2024, fielbaum_beyond_2024}, flexible buses could play a crucial role in the public transport system, especially with a high standard deviation of ridership. A further numerical analysis of the hybrid model is presented in the appendix \ref{appd:more_exps}, in which the expected ridership is 1,000 pax/h.

\begin{figure}[htbp]
    \centering
    \begin{subfigure}[htbp]{\textwidth}
        \centering
        \includegraphics[width=0.8\textwidth]{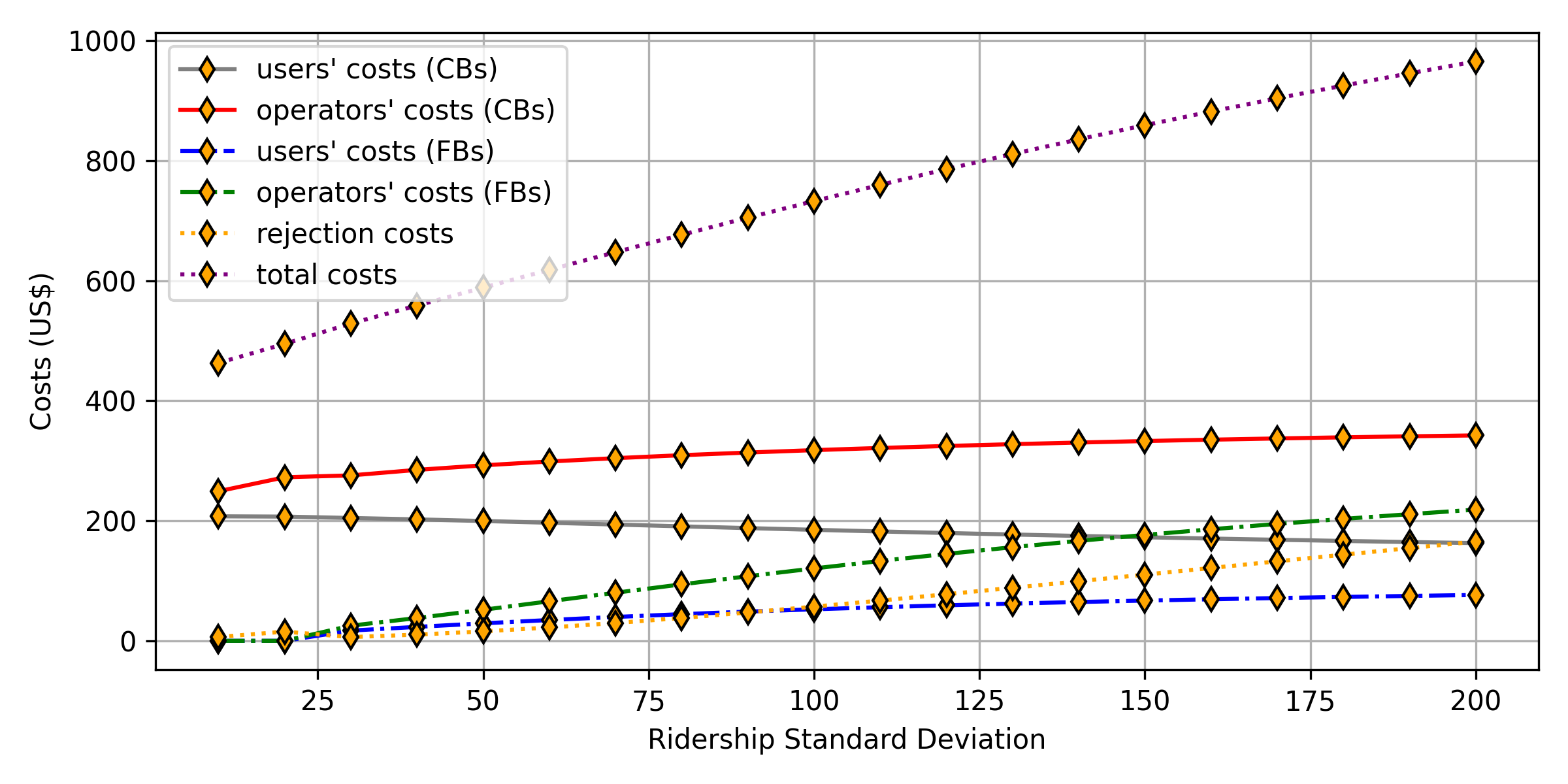}
        \caption{Breakdown of total costs under the optimal hybrid design.}
        \label{fig:HybridCosts}
    \end{subfigure}
    \vfill
    \begin{subfigure}[htbp]{\textwidth}
        \centering
        \includegraphics[width=0.8\textwidth]{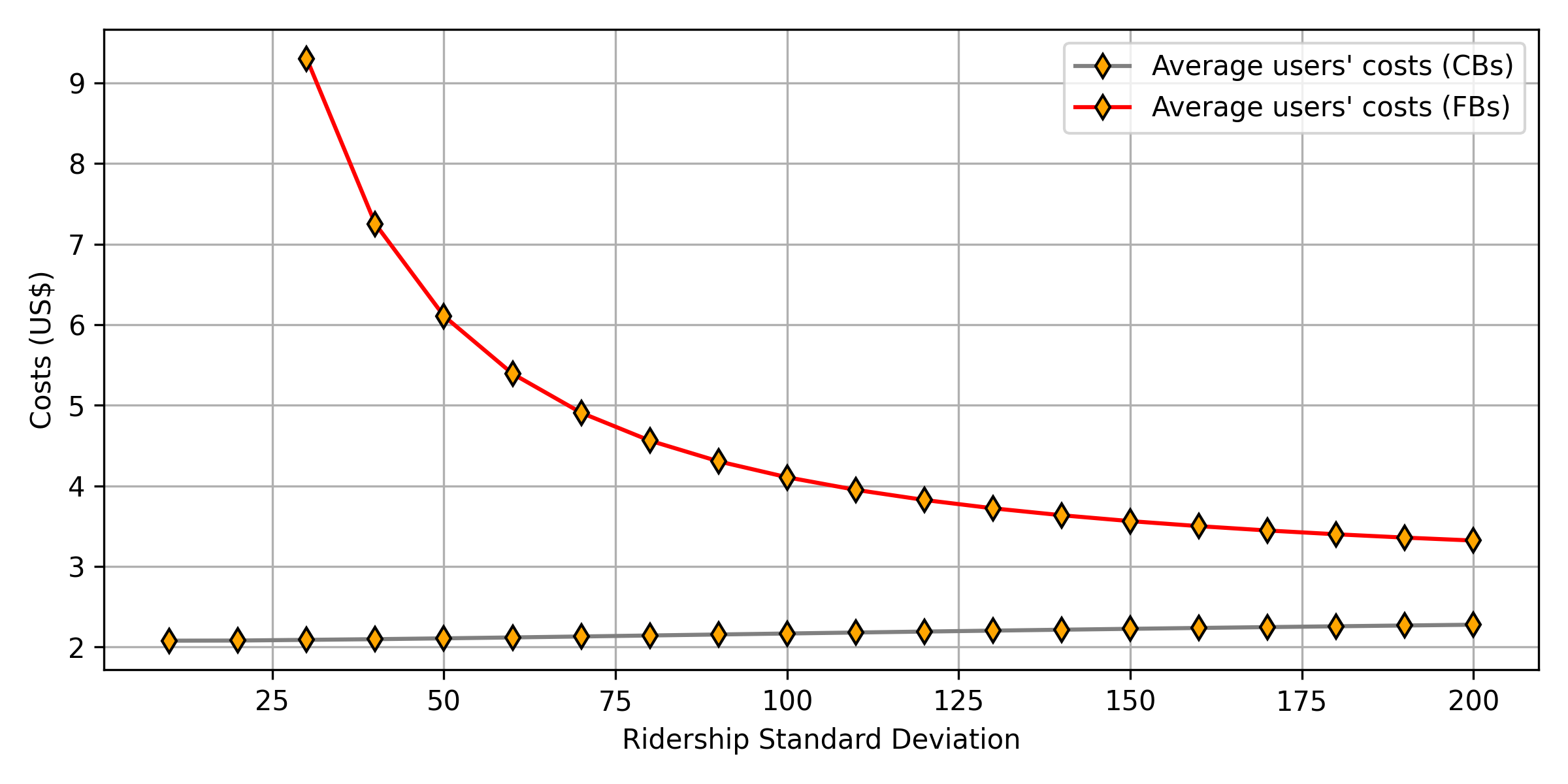}
        \caption{Average user costs per passenger under the optimal hybrid design.}
        \label{fig:HybridAvgCosts}
    \end{subfigure}
    \caption{System‐level and user‐level cost outcomes of the hybrid model under optimal design.}
    \label{fig:HybridComparison}
\end{figure}

The detailed breakdown of total costs is shown in Figure~\ref{fig:HybridCosts}. For CBs, user costs decrease slightly, whereas operators costs increase moderately as demand randomness rises to accommodate the randomness. In contrast, operator and user costs for FBs both increase with greater ridership randomness. This results from the lower costs by serving the passengers by FBs rather than rejecting them. In order to incorporate FBs, a substantial capacity of CBs is necessary to enable reasonable detours of FBs. 

Figure~\ref{fig:HybridAvgCosts} shows that after FBs are integrated into the hybrid model, their average users' costs consistently exceed those of CBs. Passengers incur higher costs when served by FBs, yet they are served rather than being rejected. As the ridership randomness grows, the average costs decline, demonstrating economics of scale.

\begin{figure}[htbp]
    \centering
    \includegraphics[width = \textwidth]{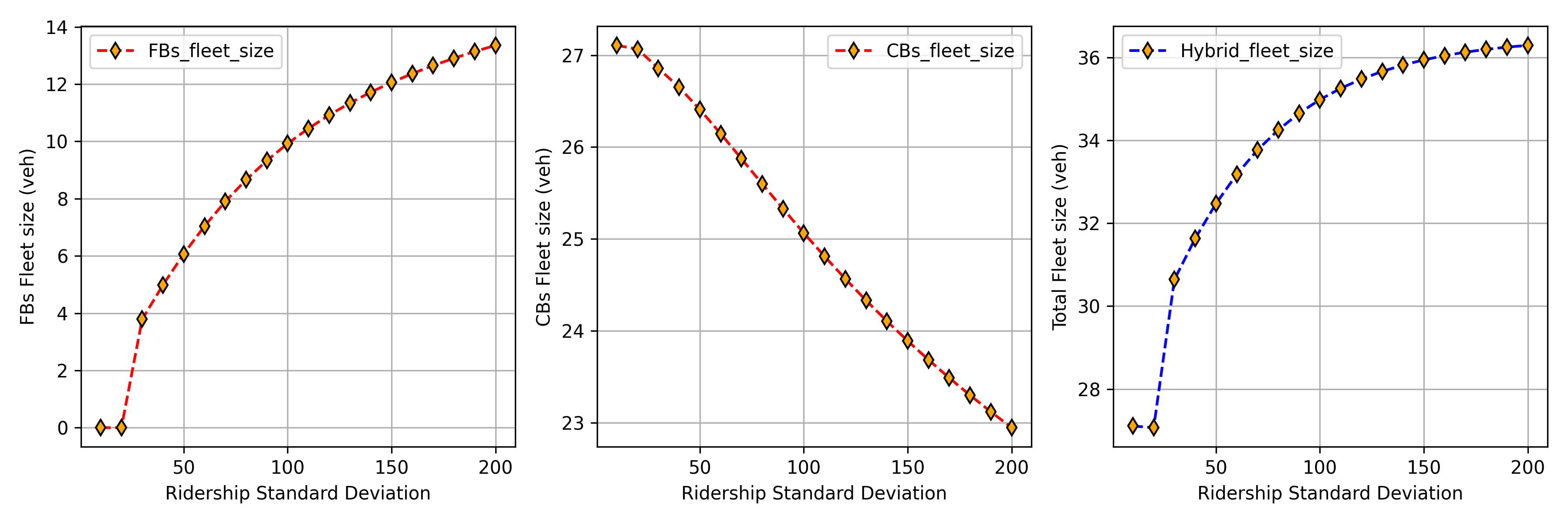}
    \caption{The fleet size of CBs, FBs and the hybrid system under increasing ridership randomness, with the mean ridership of 100 (pax/h)}
    \label{fig:TOTAL_FLEET_SIZE}
\end{figure}
Figure~\ref{fig:TOTAL_FLEET_SIZE} illustrates the fleet size of CBs, FBs, and the hybrid system. Under low ridership randomness, the fleet size of CBs remains almost constant with no FBs in the hybrid system. This aligns with Theorem \ref{Thm:Capacity}; the hybrid system accommodates the randomness with bus capacity rather than frequency. As the randomness grows, the fleet size of the hybrid system increases, which mainly results from the growth in that of FBs. As the standard deviation of ridership increases from 50 to 100, the fleet size of CBs decreases by approximately 1.5 vehicles, whereas the fleet size of FBs increases by 4 vehicles. Overall, the total fleet size increases by around 2.5 vehicles. This shift in fleet composition from CBs to FBs benefits the hybrid system by introducing more flexibility to pick up passengers underserved by CBs, rather than rejecting them.

The interaction between CBs and FBs is pronounced under different levels of the standard deviation of ridership, demonstrating how scale economics and route flexibility become dominant factors in service. Under a low standard deviation of ridership, CBs efficiently serve ridership, as expanding capacity is sufficient to reduce passenger rejections. However, as randomness increases, FBs become increasingly important in the hybrid model, utilizing route flexibility to adapt to fluctuating demand. 

\section{Summary of Main Insights and Role of the Critical Assumptions} \label{subsec:AspDis}

 We extend the single-line model to incorporate stochastic ridership and construct a multi-line model in which both CBs and FBs serve stochastic demand. From these models, we derive four main findings. We first summarize and compare them with previous studies, to establish their robustness. Later, we summarize the main assumptions that were required to determine these conclusions, and analyze their role.
\vspace{0.5em}
\noindent
 
 \textbf{Discussion of our conclusions}
 \begin{itemize}
     \item Bus capacity adjusts to accommodate ridership randomness rather than changes in frequency. This finding aligns with practical observations that operators often distribute capacity across the network to absorb sudden increases in passenger demand \parencite{cats_planning_2015}. However, such strategies are typically applied in disruption management rather than in the design of daily bus operations. Our results suggest that additional capacity should be allocated when ridership becomes more stochastic.
     \item FBs are excluded from the hybrid model when ridership randomness is low. This result aligns with insights from agent-based models \parencite{rich_fixed_2023,calabro_fixed-route_2022}. These studies emphasize that fixed-route bus services are the preferred option when passenger demand from metro or light-rail stations occurs during peak hours; while this is often related to having a larger demand volume, these periods are also less affected by randomness. We remark that it is also an intuitive result: if the demand is deterministic, there are no gains by deploying FBs, but there are costs (their detours).
     \item When the randomness increases, FBs are integrated with CBs to prevent rejections. FBs can reach passengers unserved by CBs within reasonable detour times, which reduces the total costs compared to rejecting them. With online coordination, such integration becomes feasible and has the potential to reduce overall system costs \parencite{calabro_adaptive_2023}.
     \item An efficient integration of FBs into public transport requires that the number of FBs is not too low. With few FBs, most of the ridership is served by CBs, so FBs must take unreasonable detours to pick up the remaining passengers. Such inefficiency prevents the effective integration of a small number of on-demand services. This is analogous to the ``critical mass'' concept for ridepooling \parencite{soza-parra_shareability_2024}: if ridepooling demand is small, detour penalties become high and the system is inefficient. Similarly, \textcite{fielbaum_improving_2024} ran simulations in Santiago in Berlin integrating fixed lines with ridepooling, finding that users' costs could increase if the number of flexible vehicles was too low.
 \end{itemize}
 
 These findings provide general insights, but they are illustrative rather than directly applicable, as we employ the continuous approximation approach and several simplifying assumptions are adopted in the approximation. In what follows, we examine the assumptions in detail and discuss the implications of relaxing them.
\vspace{0.5em}
\noindent
 
 \textbf{Discussion of the assumptions}
 \begin{itemize}
     \item \Assumpref{ASP:rej_pax} Passenger Rejection: We assume the passengers leave when the frequency and capacity cannot address such amount of ridership. Alternatively, we could assume that rejected passengers would wait for the next bus instead. This assumption is not critical to our results. As almost every bus has spare seats for passengers (rejection rates are very low), if a passenger decides to wait for the next bus it would leave the results unaffected.
     \item \Assumpref{ASP:pax_dis} Ridership Distribution: For both single-line and hybrid model, we assume that the cumulative distribution function of ridership decreases with its variance at the right tail. As discussed in section~\ref{sec:theo_single}, this is a natural assumption that is fulfilled in most common distributions, including Poisson, Log-normal and Gamma distributions. In our models, since the service is designed to carry passengers rather than reject them, a high rejection penalty per passenger shifts the optimal design toward the right tail of the ridership distribution. This assumption is natural, but is also critical for the results in Theorem~\ref{Thm:Capacity}, \ref{Thm:low-var-case} and \ref{Thm:high-var-case}. 
     \item \Assumpref{ASP:one_dir} One-directional Demand: In both the single-line and the hybrid model, we assume demand to be in one direction, flowing from the upstream stop to the downstream stop, with no demand in the opposite direction. This assumption could be dropped, in which case the system capacity is determined by the higher demand of high flow direction.  It is adopted here to keep the analysis tractable and transparent.
     \item \Assumpref{ASP:two_stop} Two-stop Lines: If the lines had intermediate stops, our model would no longer be tractable. New phenomena might arise, depending on the correlation between different stops. In particular, by the law of large numbers, one could expect that randomness might be mitigated, as different stops could compensate for each other's random fluctuations. Regardless, the four conclusions summarized above remain valid, but the numerical results could differ. 
     %The results could be generalized to bus lines that have more than two bus stops. This assumption could be dropped when considering the correlation between bus stations. We adopt two-stop bus lines to avoid interdependence and simplify the approximation of in-vehicle time. Otherwise, the overcrowding in the upstream bus stop could lead to irregularities (i.e. stop skipping, bus bunching).
     \item \Assumpref{ASP:Homo_line} Identical Lines: In the hybrid model, we assume that the bus lines have independent and identical design. FBs are rebalanced under this assumption, thereby ignoring spatial heterogeneity in central and bordering lines. This assumption could be dropped by replacing the homogeneous $P_C$ with that for each bus stop in Eq.~\eqref{eq:detour}. Asymmetrical lines would likely modify the threshold at which FBs appear, but the binary behavior of the system would remain the same: if there are too few FBs, the expected detour in Eq.~\eqref{eq:detour} becomes too long, regardless of whether there is a single $P_C$ or line-specific values.
     
     %However, our findings are still reasonable with this replacement. Symmetric bus lines result from spatial homogeneity, which is widely assumed in the literature. It prevents FBs from clustering along central bus lines, which could otherwise leads to imbalanced capacity and service deterioration over time. Future work could relax this assumption to account for spatial heterogeneity, and simulation-based analysis could further explore its implications.
 \end{itemize}

\section{Conclusion}
Moving beyond the deterministic ridership often assumed in the literature on strategic bus design, we introduce stochastic ridership. To accommodate such variability, we extend the single-line model with rejections to design the fleet size and bus capacity. We theoretically analyze the behavior of the optimal frequency and bus capacity under increasing variance. We find that capacity, rather than frequency, should accommodate the ridership randomness. For rising ridership, we employ greater frequency to leverage economies of scale. However, for the increase in variability, we expand capacity to absorb potential peaks and ensure robustness. 

Further, to better address the randomness, we develop a hybrid model with rejections consisting of CBs and FBs, through a simplified spatial model consisting of several parallel lines. We analyze the behavior of optimal fleet size and bus capacity of CBs and FBs with respect to the increasing standard deviation of ridership. Our findings indicate that in low-variance scenarios, it is optimal to use only CBs, thus replicating the single-line several times. Under high-variance scenarios, on the other hand,  introducing FBs is more efficient than only expanding capacity since the integration of CBs and FBs outperforms a system relying solely on CBs. Under a high standard deviation of ridership, route flexibility contributes to managing ridership randomness: buses can be dynamically rerouted to serve passengers underserved by CBs, avoiding the inefficiency of simultaneous empty seats on some lines and passenger rejections on others. 

Our paper demonstrates the value of online coordination and route flexibility in bus design, especially for stochastic ridership. Future research could extend the proposed method in several aspects. First, this is a continuous approximation model that enables a theoretical analysis but requires relevant assumptions, so agent-based simulations would be a powerful tool to identify the limitations of our results. This challenge implies some algorithmic tasks, particularly when deciding how to re-route the buses when the lines are no longer parallel.

Second, we assume homogeneous variance in the hybrid model. Diverse land use patterns across the city give rise to distinct randomness in trip requests. Additional consideration can be given to improve this assumption. For example, we could extend the model by incorporating location-dependent probability distributions to optimize dynamic rebalancing strategies, advancing beyond traditional feeder service operations \parencite{fan_optimal_2024}. The heterogeneous ridership randomness for bus design is worth exploring. For CBs, increasing bus capacity can help accommodate the stochastic demand; however, extending this approach to dynamic operation could further reduce the total costs, including short-turning and dead-heading \parencite{cortes_integrating_2011}.

\appendix
\section{ Completion of the proof of Theorem \ref{Thm:Capacity}} \label{apdx:comp_thm_cap}
\begin{proof}
   Let us begin with the first-order conditions and obtain the first-order derivatives as follows
    \begin{align}
    &\frac{\partial G}{\partial f} = -\pi_w \frac{1}{2f^2} E(X_B) + \left(\pi_w\frac{1}{2f}+\pi_v\frac{T}{2}\right) \frac{\partial E(X_B)}{\partial f} + T(\gamma_0 + \gamma_1 K) + \Pi_R \frac{\partial E(X_R)}{\partial f}\\
    &\frac{\partial G}{\partial K} = \left(\pi_w\frac{1}{2f}+\pi_v\frac{T}{2}\right) \frac{\partial E(X_B)}{\partial K} + fT\gamma_1 + \Pi_R \frac{\partial E(X_R)}{\partial K}
    \end{align}
    where the derivatives of the expected ridership with respect to $f$ are
    \begin{align}
        &\frac{\partial E(X_B)}{\partial f} = Q\mathbf{f}(Q) K + KP_R + Q(-\mathbf{f}(Q)K) = K P_R \label{eq:X_B_dev_freq}\\
        &\frac{\partial E(X_R)}{\partial f} = K\int_Q^{\infty} -\mathbf{f}(x) dx = -K P_R \label{eq:X_R_dev_freq}
    \end{align}
    Eq.~\eqref{eq:X_B_dev_freq} - \eqref{eq:X_R_dev_freq} result from the Leibniz integral rule. Similarly, derivatives of ridership with respect to $K$ are
    \begin{align}
        &\frac{\partial E(X_B)}{\partial K} = Q\mathbf{f}(Q)f + fP_R + Q(-\mathbf{f}(Q))f = f P_R \label{eq:X_B_dev_cap}\\
        &\frac{\partial E(X_R)}{\partial K} = f\int_Q^{\infty} -\mathbf{f}(x) dx = -f P_R \label{eq:X_R_dev_cap}
    \end{align}
    
    Setting the partial derivative of total costs with respect to bus capacity $K$ equal to zero, we derive the first equation from the first-order optimality conditions as
    \begin{align}
        \left(\pi_w\frac{1}{2f}+\pi_v\frac{T}{2}- \Pi_R \right)  P_R + T\gamma_1 = 0\label{eq:base1_single}
    \end{align}
    
    Setting the partial derivative of total costs with respect to $f$ equals zero yields
    \begin{align}
          &-\frac{\pi_w E(X_B)}{2f^2} +\Bigg[(\frac{\pi_w}{2f} +  \frac{\pi_vT}{2} - \Pi_R )  P_R + T \gamma_1 \Bigg]K + T\gamma_0 = 0 \label{eq:base2_single}
    \end{align}
    
    By substituting Eq.~\eqref{eq:base1_single} into Eq.~\eqref{eq:base2_single}, we obtain an equation from first-order conditions:
    \begin{align}
          & -\frac{\pi_w E(X_B)}{2f^2} + T\gamma_0 = 0  \label{eq:con_single_f}
    \end{align}
    
    By rewriting and rearranging Eq.~\eqref{eq:base1_single} and Eq.~\eqref{eq:con_single_f}, we derive equations from first-order conditions as:
    \begin{align}
        & 2T\gamma_0 f^2 = \pi_w E(X_B) \label{eq:first_order_base1}\\
        &\Omega P_R + T \gamma_1 = 0 \label{eq:first_order_base2}\\
        &\Omega = \frac{\pi_w}{2f} + \pi_v \frac{T}{2} - \Pi_R 
    \end{align}
     where $\Omega$ is the gap between the user's cost for each passenger and the rejection penalty. The result of this subtraction is negative ($\Omega < 0$), as we design the bus system to serve the ridership rather than rejecting them.
    
    Starting from what we obtain in first-order conditions, we incorporate the standard deviation of ridership into the theoretical framework to analyze how it affects the optimal solution. Introducing the variance parameter $\sigma$ into Eq.~\eqref{eq:single_obj}, we express the total costs as a function $G(f(\sigma), K(\sigma))$, where $\sigma$ represents the standard deviation of ridership. Through the first-order conditions for optimality, we establish Eq.~\eqref{eq:single_first_order_1}- Eq.~\eqref{eq:single_first_order_2}, noting that the partial derivatives have to be equal to zero when evaluated in the optimal frequency and capacity $f^*,K^*$.
    \begin{align}
    \frac{d }{d\sigma}\Big[\frac{\partial G}{\partial f} (f^*(\sigma),K^*(\sigma)) \Big]  = 0 \label{eq:single_first_order_1}\\
    \frac{d }{d\sigma}\Big[\frac{\partial G}{\partial K}(f^*(\sigma),K^*(\sigma))\Big]  = 0 \label{eq:single_first_order_2}
    \end{align}
 
    We expand Eq.~\eqref{eq:single_first_order_1}-Eq.\eqref{eq:single_first_order_2}  into Eq.~\eqref{eq:single_all_dev_sys_1} -Eq.\eqref{eq:single_all_dev_sys_2}. The solution to these equations enables us to analyze how optimal frequency and capacity respond to changes in the standard deviation of ridership.
    \begin{align}
        \frac{\partial^2 G}{\partial f^2} \frac{d f^*}{d\sigma} 
    + \frac{\partial^2 G}{\partial f \partial K} \frac{d K^*}{d\sigma} 
    + \frac{\partial^2 G}{\partial f \partial \sigma} = 0 \label{eq:single_all_dev_sys_1}\\ 
    \frac{\partial^2 G}{\partial K \partial f} \frac{d f^*}{d\sigma} 
    + \frac{\partial^2 G}{\partial K^2} \frac{d K^*}{d\sigma} 
    + \frac{\partial^2 G}{\partial K \partial \sigma} = 0 \label{eq:single_all_dev_sys_2}
    \end{align}
    
    We compute the second-order derivatives:
    \begin{align}
        & \frac{\partial^2 G}{\partial f^2} =  \pi_w f^{-3} E(X_B) - \pi_w f^{-2} \frac{\partial E(X_B)}{\partial f} + \left(\pi_w \frac{1}{2 f} + \pi_v \frac{T}{2} \right) \frac{\partial^2 E(X_B)}{\partial f^2} + \Pi_R \frac{\partial^2 E(X_R)}{\partial f^2}, \\
        &\frac{\partial^2 G}{\partial f \partial K} = \frac{\partial^2 G}{\partial K \partial f} = -\pi_w \frac{1}{2 f^2} \frac{\partial E(X_B)}{\partial K} + \left(\pi_w \frac{1}{2f} + \pi_v \frac{T}{2} \right) \frac{\partial^2 E(X_B)}{\partial f \partial K} + T\gamma_1 + \Pi_R \frac{\partial^2 E(X_R)}{\partial f \partial K} ,\\
        &\frac{\partial^2 G}{\partial K^2} = \left(\pi_w \frac{1}{2f} + \pi_v \frac{T}{2} \right) \frac{\partial^2 E(X_B)}{\partial K^2} + \Pi_R \frac{\partial^2 E(X_R)}{\partial K^2},\\
        &\frac{\partial^2 G}{\partial f \partial \sigma} = -\pi_w 
        \frac{1}{2 f^2}\frac{\partial E(X_B)}{\partial \sigma} + \left(\pi_w \frac{1}{2f} + \pi_v \frac{T}{2} \right) \frac{\partial^2 E(X_B)}{\partial f \partial \sigma} + \Pi_R \frac{\partial^2 E(X_R)}{\partial f \partial \sigma}\\
        &\frac{\partial^2 G}{\partial K \partial \sigma} =  \left( \frac{\pi_w}{2f} + \pi_v \frac{T}{2} \right) \frac{\partial^2 E(X_B)}{\partial K \partial \sigma} + \Pi_R \frac{\partial^2 E(X_R)}{\partial K \partial \sigma}
    \end{align}

    Given the derivative of ridership as Eq.~\eqref{eq:X_B_dev_freq} - Eq.~\eqref{eq:X_R_dev_cap}, we obtain the second order derivatives of ridership.
    \begin{align}
        &\frac{\partial E^2(X_B)}{\partial f^2}= -\frac{\partial^2 E(X_R)}{\partial f^2} = -K^2 \mathbf{f}(Q) \label{eq:X_B_2dev_freq}\\
        &\frac{\partial E^2(X_B)}{\partial K^2} = -\frac{\partial E^2(X_R)}{\partial K^2} = -f^2 \mathbf{f}(Q) \label{eq:X_B_2dev_cap}\\
        &\frac{\partial E^2(X_B)}{\partial f \partial K} = \frac{\partial E^2(X_B)}{\partial K \partial f} = -\frac{\partial E^2(X_R)}{\partial f \partial K} = -\frac{\partial E^2(X_R)}{\partial K \partial f} = P_R - f K \mathbf{f}(Q) \label{eq:X_B_2dev_cap_freq}
    \end{align}
    
    We simplify the second order derivatives of total costs with the equations above into
    \begin{align}
        & \frac{\partial^2 G}{\partial f^2} =  \pi_w f^{-3} E(X_B) - \pi_w f^{-2} K P_R -\Omega K^2 \mathbf{f}(Q), \\
        &\frac{\partial^2 G}{\partial f \partial K} = \frac{\partial^2 G}{\partial K \partial f} = - \frac{\pi_w}{2 f} P_R + \Omega (P_R - f K \mathbf{f}(Q)) + T\gamma_1 = -\frac{\pi_w}{2 f} P_R   - \Omega Q \mathbf{f}(Q),\\
        &\frac{\partial^2 G}{\partial K^2} = -\Omega f^2 \mathbf{f}(Q) ,\\
        &\frac{\partial^2 G}{\partial f \partial \sigma} = -\frac{\pi_w} {2 f^{2}} \frac{\partial E(X_B)}{\partial \sigma} - \Omega K \frac{\partial \mathbf{F}}{\partial \sigma}\\
        &\frac{\partial^2 G}{\partial K \partial \sigma} =  - \Omega f \frac{\partial \mathbf{F}}{\partial \sigma}
    \end{align}
    
    Given the second-order derivatives above, the solution to the Eq.~\eqref{eq:single_first_order_1} and Eq.~\eqref{eq:single_first_order_2} is 
    \begin{align}
        &\frac{d f^*}{d\sigma} = \frac{\frac{\pi_w}{2}\Omega \left[P_R \frac{\partial F}{\partial \sigma} -  \frac{\partial E(X_B)}{\partial \sigma}\mathbf{f}(Q)\right]}{\pi_w(-\Omega) f^{-1} E(X_B) \mathbf{f}(Q) - \frac{\pi_w^2}{4} f^{-2} P_R^2} \label{eq:single_opt_freq_sigma}\\
        &\frac{d K^*}{d\sigma} = \frac{\frac{\pi_w^2}{4} f^{-3} P_R \frac{\partial E(X_B)}{\partial \sigma}+\pi_w \Omega f^{-2} E(X_B) \frac{\partial \mathbf{F}}{\partial \sigma} - \frac{\pi_w}{2}\Omega f^{-1} K  \left[P_R \frac{\partial F}{\partial \sigma} -  \frac{\partial E(X_B)}{\partial \sigma}\mathbf{f}(Q)\right]}{\pi_w(-\Omega) f^{-1} E(X_B) \mathbf{f}(Q) - \frac{\pi_w^2}{4} f^{-2} P_R^2} \label{eq:single_opt_cap_sigma}
    \end{align}
    
    Eq.~\eqref{eq:single_opt_freq_sigma}-\eqref{eq:single_opt_cap_sigma} have the same denominator. The denominator is positive when the rejection penalty is significant. For the numerator, both equations depend on the term:
    \begin{align}
        \left[P_R \frac{\partial F}{\partial \sigma} -  \frac{\partial E(X_B)}{\partial \sigma}\mathbf{f}(Q)\right] \label{eq:common_term_dev}
    \end{align}
    According to Eq.~\eqref{eq:first_order_base2}, when the rejection penalty is sufficiently large, $P_R$ approaches zero, as well as $f(Q)$. This leads to the disappearance of Eq.~\eqref{eq:common_term_dev}.  Consequently, it causes the numerator in Eq.~\eqref{eq:single_opt_freq_sigma} to approach zero. This concludes that \textbf{ the optimal frequency is maintained with an increase in the standard deviation of ridership}. 
    
    Given a sufficiently large penalty in $\Omega$,  the first term in the numerator and the second term in the denominator of Eq.~\eqref{eq:single_opt_cap_sigma} are negligible. The third term in the numerator approaches zero, as explained in the analysis of the behavior of the optimal frequency. Therefore, it can be simplified as 
    \begin{align}
        \frac{-\frac{\partial \mathbf{F}}{\partial \sigma}}{f \mathbf{f}(Q)}
    \end{align}
    In the equation above, the $-\frac{\partial F}{\partial \sigma}$ is positive according to Assumption(\ref{ASP:pax_dis}), and $\mathbf{f}(Q)$ approaches zero when the rejection penalty is significant. These result in an increase in the optimal bus capacity with the standard deviation of ridership. This concludes that \textbf{capacity accommodates the standard deviation of ridership}. 
    \end{proof}

\section{ Completion of the proof of Theorem \ref{Thm:high-var-case}}
\label{apdx:comp_thm_high}
\begin{proof}
     Let us start with the total costs of the hybrid model, which is 
     \begin{align}
         &\text{min} \quad  \underbrace{\left( \frac{\pi_w}{2} \omega\frac{N(2L_x + D)}{s_F V} + \frac{\pi_v  (L_x + \frac{D}{2}) }{V} \right) E(X_F)} _{\text{Users' Costs of FBs}}
        + \underbrace{s_F(\gamma_0 + \gamma_1 K_F)}_{\text{Operators' Costs of FBs}}  \nonumber\\
        &\quad +\underbrace{ \Bigl(\frac{\pi_w}{2f_C} +   \frac{\pi_v L_x}{V} \Bigr)E(X_C)}_{\text{Users' Costs of CBs}}
        +\underbrace{N(2f_C\frac{L_x}{V} )(\gamma_0 + \gamma_1 K_C)}_{\text{Operators' Costs of CBs}}
        +\underbrace{E(X_R)\Pi_{REJ}}_{\text{Rejection Costs}}
     \end{align}
     
     We determine the optimal fleet size of FBs using first-order conditions. Derivatives of FBs costs with respect to $s_F$ and $K_F$ are
\begin{align}
    &-\frac{\pi_w}{2} \omega\frac{N(2L_x + D)}{s_F^2 V} E(X_F) + \frac{\pi_w}{2} \omega\frac{N(2L_x + D) }{s_F V} \frac{\partial E(X_F)}{\partial s_F} + \frac{\pi_v (L_x + \frac{D}{2})}{V} \frac{\partial E(X_F)}{\partial s_F}  + (\gamma_0 + \gamma_1 K_F)\\
    &\left( \frac{\pi_w}{2} \omega \frac{N(2L_x + D)}{s_F V} + \frac{\pi_v  (L_x + \frac{D}{2}) }{V} \right) \frac{\partial E(X_F)}{\partial K_F} + s_F \gamma_1
\end{align}
The derivative of rejection cost with respect to $s_F$ and $K_F$ are $ \Pi_R \frac{\partial E(X_R)}{\partial s_F}$ and $ \Pi_R \frac{\partial E(X_R)}{\partial K_F}$. The ridership of FBs and rejected is
\begin{align}
    &E(X_F) = \int_{Q_{CB}}^{Q_{HY}} (x - Q_{CB}) \mathbf{f}(x) dx + P_R Q_{FB}  \\   
    & E(X_R) = \int_{Q_{HY}}^{\infty} (x -Q_{HY}) \mathbf{f}(x) dx 
\end{align}
Derivatives of ridership with respect to the $s_F$ and $K_F$ are
\begin{align}
    &\frac{\partial E(X_F)}{\partial s_F} =  K_F \frac{V}{\omega (2L_x + D)} \left(1 + \frac{D}{L_y} \right) P_R \label{eq:dev_serve_s_F_rid_FB}\\
    &\frac{\partial E(X_R)}{\partial s_F} = -K_F \frac{V}{\omega(2L_x + D)} \left(1 + \frac{D}{L_y}\right)  P_R \label{eq:dev_rej_s_F_rid_FB}\\
    &\frac{\partial E(X_F)}{\partial K_F} =  s_F\frac{ V}{\omega(2L_x + D)} \left(1 + \frac{D}{L_y}\right) P_R \label{eq:dev_serve_s_F_FB}\\
    &\frac{\partial E(X_R)}{\partial K_F} = - s_F \frac{ V}{\omega(2L_x + D)} \left(1 + \frac{D}{L_y}\right) P_R \label{eq:dev_rej_s_F_FB}
\end{align}
Under first-order conditions, we have 
\begin{align}
    &-\frac{\pi_w}{2} \omega\frac{N(2L_x + D)}{s_F^2 V} E(X_F) + \Omega_F \frac{\partial E(X_F)}{\partial s_F}  + (\gamma_0 + \gamma_1 K_F)  = 0 \label{eq:thm_high_first_order_s_f}\\
    &\Omega_F \frac{\partial E(X_F)}{\partial K_F} + s_F \gamma_1 = 0 \label{eq:thm_high_first_order_K_f}\\
    &\Omega_F =  \frac{\pi_w}{2} \omega \frac{N(2L_x + D)}{s_F V} + \frac{\pi_v  (L_x + \frac{D}{2}) }{V} - \Pi_R \\
    &\frac{\partial E(X_F)}{\partial K_F} \frac{K_F}{s_F} = \frac{\partial E(X_F)}{\partial s_F}
\end{align}
Replacing Eq.\eqref{eq:thm_high_first_order_K_f} with Eq.\eqref{eq:dev_serve_s_F_FB} yields 
\begin{align}
  &\left(\frac{\pi_w}{2} \omega\frac{N(2L_x + D)}{s_F V} + \frac{\pi_v  (L_x + \frac{D}{2}) }{V} - \Pi_R \right) \frac{ V}{\omega(2L_x + D)} \left(1 + \frac{D}{L_y}\right) P_R +  \gamma_1 = 0  \label{eq:OPT_S_F_1}
\end{align}
Rearranging Eq.\eqref{eq:OPT_S_F_1}, we obtain
\begin{align}
  & s_F^* = \frac{\frac{\pi_w N}{2} \left(1 + \frac{D}{L_y}\right) P_R}{(\frac{\Pi_R V}{\omega(2L_x + D)} -  \frac{\pi_v}{2}) \left(1 + \frac{D}{L_y}\right) P_R - \gamma_1} \label{eq:OPT_S_F_2}
\end{align}

Now we have the formula of the optimal fleet size for FBs. Let us study the derivative of the total costs with respect to $f_C$ and $K_C$. The derivative of CBs costs with respect to $f_C$ is 
\begin{align}
     \frac{-\pi_w}{2f_C^2} E(X_C) + \left(\frac{\pi_w}{2f_C}+\frac{\pi_v L_x}{2 V}\right) \frac{\partial E(X_C)}{\partial f_C} + 2 N \frac{L_x}{V}(\gamma_0 + \gamma_1 K_C) 
\end{align}
The derivative of FBs costs with respect to $f_C$ is 
\begin{align}
    \left( \frac{\pi_w N}{2 s_F V}  + \frac{\pi_v}{2 V} \right)E(X_F) \frac{\partial D}{\partial f_C}  + \left( \frac{\pi_w}{2} \omega\frac{N(2L_x + D)}{s_F V} + \frac{\pi_v  (L_x + \frac{D}{2}) }{V} \right) \frac{\partial E(X_F)}{\partial f_C}
\end{align}
where 
\begin{align}
    & \frac{\partial D}{\partial f_C} = L_y N K_C\frac{1 + \mathbf{F}^2(Q_{CB})  }{(1 - \mathbf{F}^2(Q_{CB}))^2 }  \mathbf{f}(Q_{CB})  \\
    &\frac{\partial E(X_C)}{\partial f_C} =  N K_C P_C\\
    &\frac{\partial E(X_F)}{\partial f_C} = -N K_C \left( P_C - P_R \right)
\end{align}

The derivative of rejection costs with respect to $f_C$ is 
\begin{align}
    \Pi_R (-N K_C) P_R
\end{align}

The derivative of CBs costs with respect to $K_C$ is 
\begin{align}
    \Bigl(\frac{\pi_w}{2f_C} +   \frac{\pi_v L_x}{V} \Bigr)\frac{\partial E(X_C)}{\partial K_C} + 2 N f_C \frac{L_x}{V} \gamma_1
\end{align}

The derivative of FBs costs with respect to $K_C$ is 
\begin{align}
    \left( \frac{\pi_w N}{2 s_F V}  + \frac{\pi_v}{2 V} \right)E(X_F) \frac{\partial D}{\partial K_C}  + \left( \frac{\pi_w}{2} \omega\frac{N(2L_x + D)}{s_F V} + \frac{\pi_v  (L_x + \frac{D}{2}) }{V} \right) \frac{\partial E(X_F)}{\partial K_C}
\end{align}
where
\begin{align}
    &\frac{\partial D}{\partial K_C} =  L_y N f_C \frac{1 + \mathbf{F}^2(Q_{CB})  }{(1 - \mathbf{F}^2(Q_{CB}))^2 }  \mathbf{f}(Q_{CB})\\
    &\frac{\partial E(X_C)}{\partial K_C}  = N f_C P_C\\
    & \frac{\partial E(X_F)}{\partial K_C} = -N f_C \left( P_C - P_R \right)
\end{align}

The relation between the derivatives of detours $D$ and $E(X_F)$ with respect to $K_C$ and $f_C$ are
\begin{align}
    &\frac{\partial D}{\partial K_C} = \frac{f_C}{K_C} \frac{\partial D}{\partial f_C} \label{eq:ratio_D}\\
    &\frac{\partial E(X_C)}{\partial K_C} = \frac{f_C}{K_C} \frac{\partial E(X_C)}{\partial f_C}  \label{eq:ratio_X_C} \\
    &\frac{\partial E(X_F)}{\partial K_C}  = \frac{f_C}{K_C} \frac{\partial E(X_F)}{\partial f_C} \label{eq:ratio_X_F}
\end{align}

The derivative of rejection costs with respect to $K_C$ is 
\begin{align}
    \Pi_R (-N f_C) P_R
\end{align}

Setting the derivative of the total costs with respect to $f_C$ and $K_C$ to zero, we obtain
\begin{align}
    &\frac{-\pi_w}{2f_C^2} E(X_C) + \left(\frac{\pi_w}{2f_C}+\frac{\pi_v L_x}{2 V}\right) \frac{\partial E(X_C)}{\partial f_C} + 2 N \frac{L_x}{V}(\gamma_0 + \gamma_1 K_C) +\left( \frac{\pi_w N}{2 s_F V}  + \frac{\pi_v}{2 V} \right)E(X_F) \frac{\partial D}{\partial f_C}  \notag \\
    &+ \left( \frac{\pi_w}{2}\omega \frac{N(2L_x + D)}{s_F V} + \frac{\pi_v  (L_x + \frac{D}{2}) }{V} \right) \frac{\partial E(X_F)}{\partial f_C} + \Pi_R (-N K_C) P_R = 0 \label{eq:dev_f_C_hybrid_1} \\
    \notag \\
    &\Bigl(\frac{\pi_w}{2f_C} +   \frac{\pi_v L_x}{V} \Bigr) \frac{\partial E(X_C)}{\partial K_C} + 2 N f_C \frac{L_x}{V} \gamma_1 + \left( \frac{\pi_w N}{2 s_F V}  + \frac{\pi_v}{2 V} \right) E(X_F) \frac{\partial D}{\partial K_C}   \notag \\
    &+ \left( \frac{\pi_w}{2} \omega \frac{N(2L_x + D)}{s_F V} + \frac{\pi_v  (L_x + \frac{D}{2}) }{V} \right) \frac{\partial E(X_F)}{\partial K_C} + \Pi_R (-N f_C) P_R = 0 \label{eq:dev_K_C_hybrid_1}
\end{align}

Substituting Eq.~\eqref{eq:ratio_D} - Eq.~\eqref{eq:ratio_X_F} into Eq.~\eqref{eq:dev_K_C_hybrid_1}, we obtain
\begin{align}
    &\Bigl(\frac{\pi_w}{2f_C} +   \frac{\pi_v L_x}{V} \Bigr) \frac{f_C}{K_C} \frac{\partial E(X_C)}{\partial f_C} + 2 N f_C \frac{L_x}{V} \gamma_1 + \left( \frac{\pi_w N}{2 s_F V}  + \frac{\pi_v}{2 V} \right) E(X_F) \frac{f_C}{K_C}\frac{\partial D}{\partial f_C}   \notag \\
    &+ \left( \frac{\pi_w}{2} \omega\frac{N(2L_x + D)}{s_F V} + \frac{\pi_v  (L_x + \frac{D}{2}) }{V} \right) \frac{f_C}{K_C}\frac{\partial E(X_F)}{\partial f_C} + \Pi_R (-N f_C) P_R = 0 
\end{align}
which can be simplified as 
\begin{align}
    &\Bigl(\frac{\pi_w}{2f_C} +   \frac{\pi_v L_x}{V} \Bigr)  \frac{\partial E(X_C)}{\partial f_C} + \left( \frac{\pi_w N}{2 s_F V}  + \frac{\pi_v}{2 V} \right) E(X_F) \frac{\partial D}{\partial f_C}   \notag \\
    &+ \left( \frac{\pi_w}{2}\omega \frac{N(2L_x + D)}{s_F V} + \frac{\pi_v  (L_x + \frac{D}{2}) }{V} \right) \frac{\partial E(X_F)}{\partial f_C} =   K_C \left( \Pi_R N P_R - 2 N  \frac{L_x}{V} \gamma_1 \right) \label{eq:dev_K_C_hybrid_2}
\end{align}
Replacing Eq.~\eqref{eq:dev_f_C_hybrid_1} with Eq.~\eqref{eq:dev_K_C_hybrid_2} yields
\begin{align}
    &\frac{-\pi_w}{2f_C^2} E(X_C) + 2 N \frac{L_x}{V}(\gamma_0 + \gamma_1 K_C) + K_C \left( \Pi_R N P_R - 2 N  \frac{L_x}{V} \gamma_1 \right) - \Pi_R N K_C P_R = 0\\
    &\frac{-\pi_w}{2f_C^2} E(X_C) + 2 N \frac{L_x}{V}\gamma_0  = 0
\end{align}

which can be rearranged as 
\begin{align}
     f_C^2   = \frac{\pi_w  E(X_C)}{4 N \frac{L_x}{V}\gamma_0}
\end{align}

Let us analyze the behavior of the optimal fleet size of CBs and FBs, which are 
\begin{align}
    & s_F^* = \frac{\frac{\pi_w N}{2} }{\frac{\Pi_R V}{\omega(2L_x + D)} - \frac{\pi_v}{2} - \frac {\gamma_1}{\left(1 + \frac{D}{L_y}\right) P_R}} \label{eq:OPT_S_F_FINAL} \\
    & f_C^* = \sqrt{\frac{\pi_w V  E(X_C)}{4 N L_x \gamma_0}}
\end{align}

The probability of being underserved by CBs is bounded by the probability of rejections, as 
\begin{align}
    & Q_{CB} \leq Q_{HY} \\
    & P_C = 1-\mathbf{F}(Q_{CB}) \geq 1-\mathbf{F}(Q_{HY}) = P_R > 0
\end{align}

  $Q_{CB} = 0$ if and only if $P_C = 1 - \mathbf{F}(Q_{CB}) = 0$. However, $P_C \geq P_R > 0$, then $Q_{CB} > 0$. It can be concluded that we always have CBs in the hybrid model when there are rejections. \\
  According to Corollary \ref{Coro:low-var-1}, when ridership randomness is low and detours are excessive, Eq.~\eqref{eq:OPT_S_F_FINAL} yields a negative value for the optimal fleet size of FBs. However, due to the non-negative constraint, the actual optimal fleet size is zero. When the ridership randomness is high, detours are no longer excessive, and the rejection penalty is significant enough to make the denominator of the right-hand side of Eq.~\eqref{eq:OPT_S_F_FINAL} positive. Since the numerator of the right-hand side of Eq.~\eqref{eq:OPT_S_F_FINAL} is positive, the optimal fleet size of FBs is positive. \\
  In conclusion, with large ridership randomness, the optimal fleet sizes of CBs and FBs are positive.

\end{proof}

\section{ Further numerical experiments}
\label{appd:more_exps}
The optimal design performs similarly, with the standard deviation increasing from 0 to 1,000 (pax/h) and the mean ridership remaining at 1,000 (pax/h). The results are illustrated in Figure~\ref{fig:OP-INPUT-SINGLE-1000}. The ridership follows a log-normal distribution. The optimal frequency stays almost constant with the standard deviation, which decreases by approximately 0.5 vehicles with a standard deviation from 0 to 1,000 passengers per hour. This results from the rejection probability increasing with the ridership randomness, and fewer passengers are served by CBs. To maintain the systematic capacity to prevent rejections, the optimal bus capacity increases with the ridership standard deviation. Given a substantial rejection penalty and the increase in the rejected passengers, the total costs of the extended single-line model increase with the ridership randomness.
\begin{figure}[htbp]
    \centering
    \includegraphics[width = \textwidth]{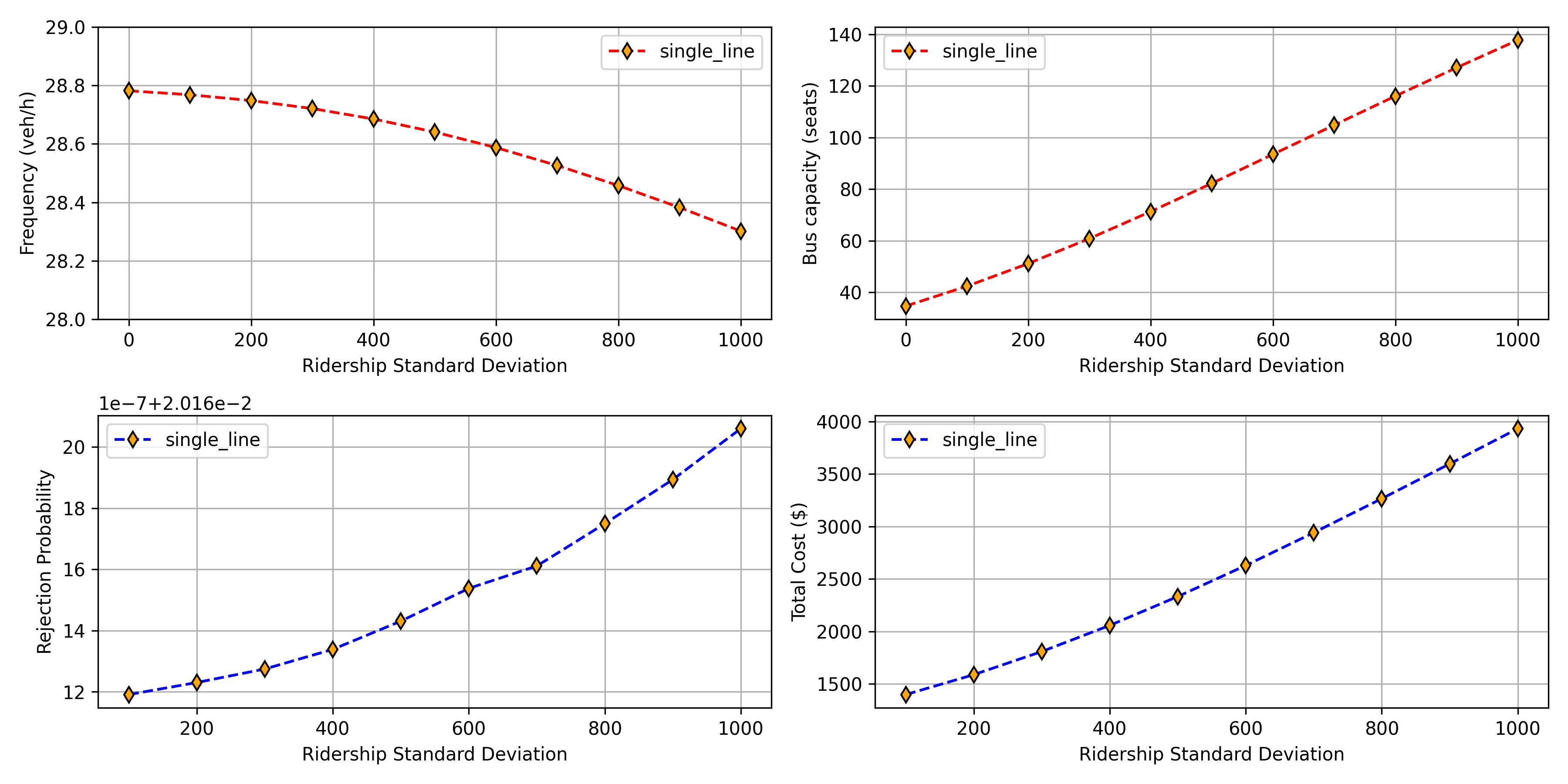}
    \caption{Optimal design variables for the extended single-line model with mean ridership of 1,000 passengers per hour}
    \label{fig:OP-INPUT-SINGLE-1000}
\end{figure}

The hybrid model has a similar behavior of the optimal frequency and bus capacity in the experiment, of which the mean ridership is 1,000 passengers per hour. As shown in Fig~\ref{fig:OP-INPUT-Hybrid-1000}, the frequency of CBs stays almost constant with the ridership variance. This agrees with the theoretical analysis that the frequency of CBs is determined by the served ridership. Given the optimal frequency is determined by the served ridership, the hybrid system increases the bus capacity of CBs for ridership randomness at first. However, the integration of FBs is more economical than increasing capacity when the ridership randomness is high. The system operates in a binary state: either CBs serve all ridership, or CBs handle the majority, while FBs address the minority portion. Our hybrid bus system outperforms the benchmark model, as shown in Figure.~\ref{fig:OP-COM-Hybrid-1000}.

\begin{figure}[htbp]
    \centering
    \includegraphics[width = \textwidth]{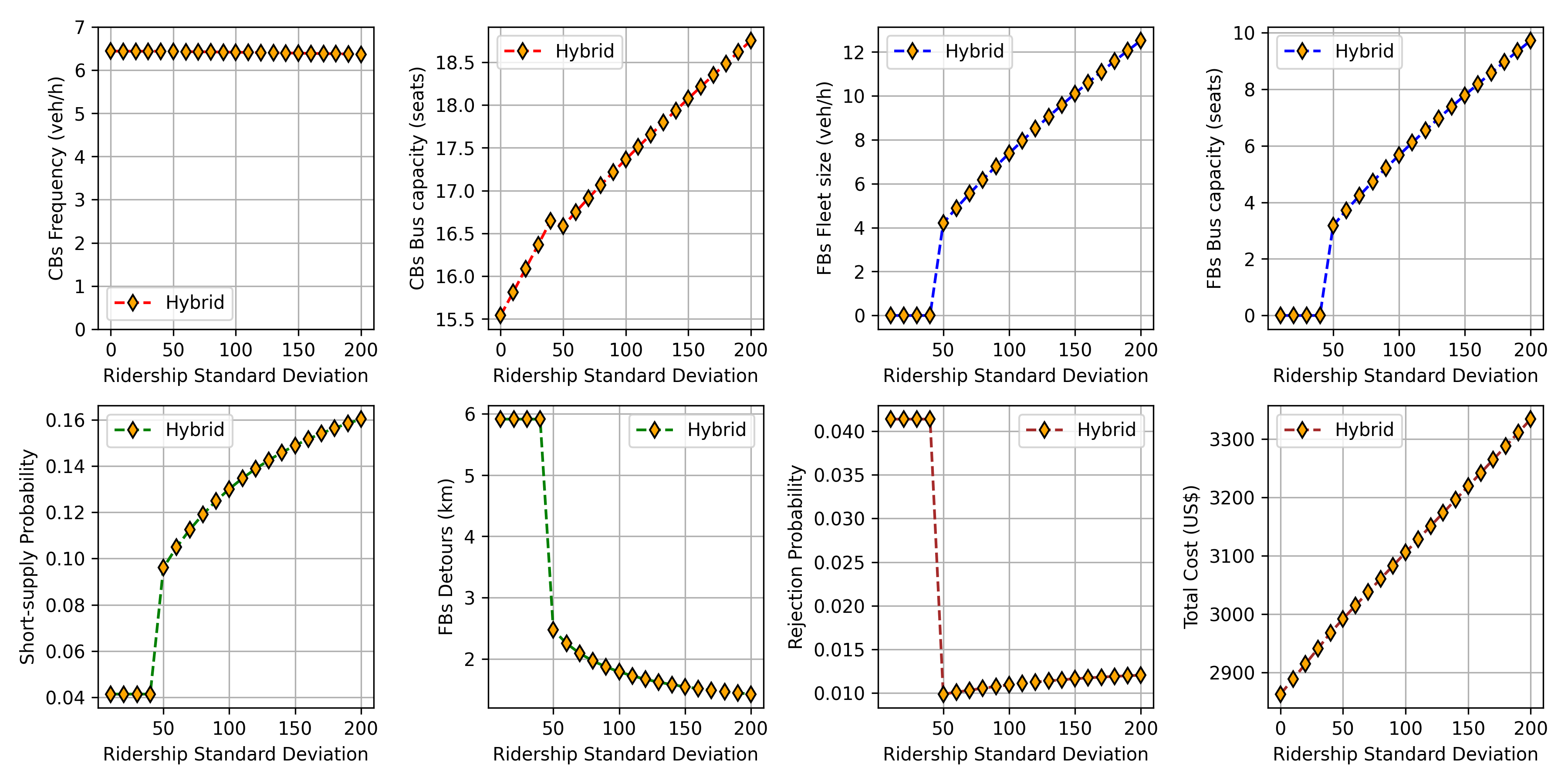}
    \caption{Optimal design variables for the hybrid model with mean ridership of 1,000 passenger per hour}
    \label{fig:OP-INPUT-Hybrid-1000}
\end{figure}

\begin{figure}[htbp]
    \centering
    \includegraphics[width = 0.8\textwidth]{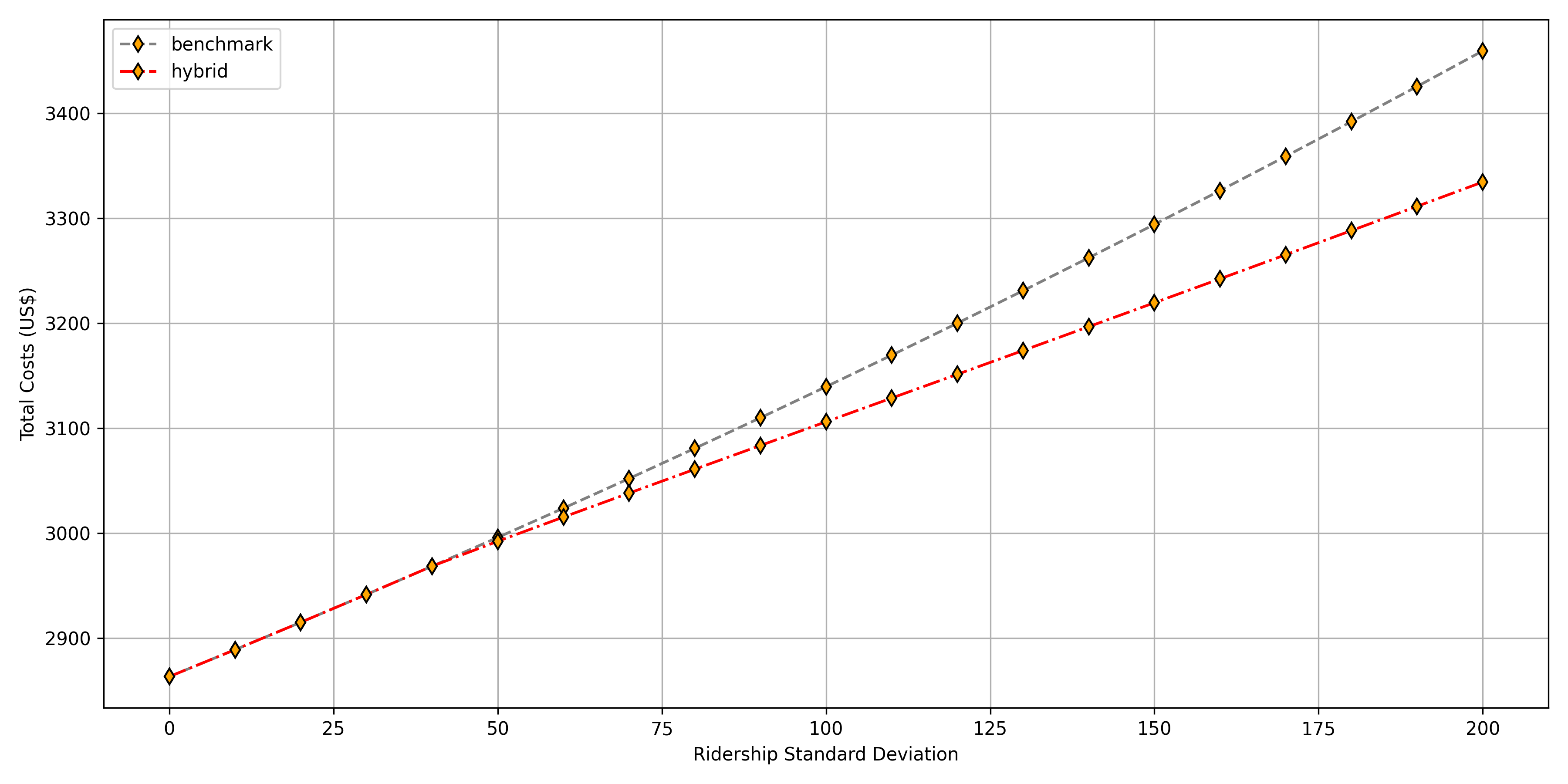}
    \caption{Comparison of the hybrid model to the benchmark model with mean ridership of 1,000 passenger per hour}
    \label{fig:OP-COM-Hybrid-1000}
\end{figure}

\newpage
\printbibliography

\end{document}